\documentclass[12pt,a4paper]{article}
\usepackage{eurosym}
\usepackage{geometry}
\usepackage{amssymb}
\usepackage{amsmath,bm}
\usepackage{amsfonts}
\usepackage{graphicx,epsfig,lscape}
\usepackage{subcaption}
\usepackage{color}
\usepackage[dvipsnames]{xcolor}
\usepackage{setspace}
\usepackage{url}
\usepackage{upgreek}
\usepackage{commath}
\usepackage{enumitem}
\usepackage{indentfirst}
\usepackage{hyperref}
\hypersetup{
  colorlinks,
  citecolor=Blue,
  linkcolor=Maroon,
  urlcolor=Blue}
\usepackage{natbib}
\usepackage{booktabs}
\usepackage{pifont}
\usepackage{cleveref}
\usepackage{fnpct}
\usepackage{titling}
\usepackage{multibib}
\usepackage{cleveref}
\usepackage{arydshln}
\usepackage{caption}
\usepackage{bbm}
\usepackage{extarrows}
\usepackage{xcolor}
\usepackage{amsthm}
\usepackage{comment}
\usepackage{pgfplots}
\usetikzlibrary{patterns,intersections}
\usetikzlibrary{pgfplots.fillbetween}
\pgfplotsset{width=7cm,compat=1.16}
\crefname{assumption}{assumption}{assumptions}

\usepackage{kpfonts}

\newtheorem{assumption}{Assumption}

\newtheorem{claim}{Claim}

\newtheorem{corollary}{Corollary}

\newtheorem{lemma}{Lemma}

\newtheorem{proposition}{Proposition}
\newtheorem{remark}{Remark}

\usepackage{chngcntr}
\usepackage{apptools}
\AtAppendix{\counterwithin{equation}{section}}
\AtAppendix{\counterwithin{figure}{section}}
\AtAppendix{\counterwithin{table}{section}}
\AtAppendix{\counterwithin{lemma}{section}}
\AtAppendix{\counterwithin{remark}{section}}
\AtAppendix{\counterwithin{definition}{section}}
\AtAppendix{\counterwithin{proposition}{section}}
\AtAppendix{\counterwithin{theorem}{section}}
\AtAppendix{\counterwithin{corollary}{section}}

\newcommand{\precision}{\rho}

\title{Pitfall of Precision in Noisy Signaling\thanks{%
We would like to thank Alessandra Casella, Navin Kartik, Qingmin Liu, Jacopo Perego, and the audience at Columbia University Micro Theory Colloquium for their valuable comments and suggestions.}}
\author{Shuhua Si\thanks{%
Google. Email: \href{mailto:ss5580@columbia.edu}%
{ss5580@columbia.edu}.}
\ and Yangfan Zhou\thanks{%
Department of Economics, Columbia University. Email: \href{mailto:yz3905@columbia.edu}%
{yz3905@columbia.edu}.}}
\date{May 1, 2026}

\begin{document}
\maketitle
\begin{abstract}

A principal decides whether to approve an agent based on a noisy signal (e.g., test scores) generated by the agent. High-quality agents can produce high signals on average at lower cost, but the realizations are subject to noise that depends on the screening technology's precision. We uncover a paradoxical ``pitfall of precision'': when precision is already high, further improvements reduce screening accuracy and lower the principal's welfare. This occurs because greater precision incentivizes strategic signaling from more low-quality agents, outweighing the direct benefit from improved precision. The pitfall of precision also has implications for statistical discrimination: groups with noisier technologies face lower approval rates yet may be favored ex ante---a reversal of discrimination. We also examine how commitment power helps mitigate the pitfall.


\newpage
\end{abstract}


\newpage

\section{Introduction}
Firms screen job applicants with résumés and interviews; banks assess borrowers through credit scores; universities admit students based on test scores. Across these settings, institutions rely on noisy signals to infer candidates' types and select eligible ones. Since noise obscures information, one might expect institutions to prefer more precise signals. 
Yet many institutions choose noisier methods deliberately, even when more precise evaluation is technologically feasible. 
For example, in college admissions, 
test-optional policies 
have expanded 
even as standardized testing has grown more sophisticated. Why would a university choose to blur information rather than sharpen it? 

One 
concern motivating these choices could be that tests invite strategic responses: intensive preparation, repeated retaking, and expensive coaching that may undermine the informational value of tests. Similarly, in recruitment, objective technical screens that can be scored precisely (e.g., coding tasks
) may invite intensive preparation (e.g., LeetCode grinding
) that undermines the screening of on-the-job performance. 

Such strategic responses to screening technologies are pervasive. Professionals devote resources to narrowly passing licensing exams; firms bunch just below audit thresholds to avoid stricter tax enforcement \citep{onji2009response, almunia2018under}; households rearrange assets to qualify for Medicaid long-term care \citep{bassett2007medicaid}. Crucially, these strategic behaviors depend on the screening technology.

This paper studies the interaction between screening precision and strategic responses. We investigate whether high precision can be undesirable for screening institutions precisely because it intensifies strategic behavior. As a benchmark, absent such behavior, higher precision unambiguously improves screening: signals become more accurate, and the institution's expected payoff rises \citep{blackwell1951comparison, lehmann1988comparing}. With strategic behavior, however, two opposing forces arise. On the one hand, greater precision reduces random errors (direct effects): fewer agents who exert little effort pass by luck, and fewer who invest substantial effort fail due to bad luck. On the other hand, precision raises the value of signaling: agents who invest costly effort to appear qualified 
can be more confident that doing so will lead to approval. 
This encourages more strategic responses (indirect, strategic effect). 

Whether this strategic effect is harmful, and if so, whether it can be strong enough to overturn the direct benefits of precision, depends on the environment.\footnote{Indeed, there exist environments in which the direct effects always dominate; see \Cref{rm:extensive}.} We show that approval-rejection environments generate non-monotone marginal benefits from effort, creating an extensive margin for the strategic effect: more unqualified agents are encouraged to exert effort. At high precision, this extensive margin causes the strategic effect to dominate, leaving the institution worse off. Echoing \citeauthor{goodhart1975problems}'s Law, ``when a measure becomes a target, it ceases to be a good measure'', we call this paradox the \textit{pitfall of precision}.

Our results thus offer a rationale for why screening institutions might deliberately limit precision in practice. By down-weighting standardized tests, e.g., through holistic admissions, universities could dampen strategic test-preparation rat races and improve screening efficiency.\footnote{For test-optional policies specifically, \cite{dessein2025test} suggest these may also reflect colleges' responses to social pressure on admission decisions.} Relatedly, candidates from disadvantaged groups, such as those with language barriers 
or fewer resources to prepare for tests, tend to face noisier evaluations. The pitfall of precision thus also speaks to why institutions may wish to direct more favorable quotas or attention toward such groups.

To formalize these ideas, we study a noisy signaling model with continuous types. An agent with privately known type $\theta\in[\underline\theta,\overline\theta]$ applies to a principal, signaling his type by choosing effort $e$ at a linear cost, where higher types face lower marginal costs. The principal observes a noisy signal of the agent's effort: $s=e+\epsilon/\precision$, where $\epsilon$ is a zero-mean noise and $\precision>0$ is the precision level.\footnote{The setup is isomorphic to one where signals depend on both types and effort complementarily; see \Cref{sect:model discussion}.} The principal then decides whether to approve the application: approving higher types yields higher payoffs, and only types above a threshold $\tilde\theta$ yield positive payoffs. We refer to types above $\tilde\theta$ as \emph{good types} and those below as \emph{bad types}.\footnote{Meanwhile, we use \emph{high} and \emph{low} to indicate relative positions of types.} We assume that, absent signals, the principal would reject the application. 

Proposition \autoref{prop:eqm} shows that when precision is high, two equilibria exist: (1) a pooling equilibrium, in which all agents exert zero effort and the principal rejects all applications; and (2) a semi-separating equilibrium, in which the principal sets a positive approval standard and types above an endogenous threshold $\hat\theta$ exert positive effort (above the standard) while types below exert zero effort. Crucially, $\hat\theta<\tilde\theta$: all good types exert effort and some bad types are incentivized to mimic them. Other equilibria may exist for fixed precision levels, but only these two survive in the vanishing noise limit ($\precision\to\infty$) (Proposition \autoref{prop:limit}). 

Our main results concern how screening precision $\precision$ affects outcomes in the semi-separating equilibrium. In particular, we show the principal's payoff in this equilibrium is non-monotone in precision (Proposition \autoref{prop:non-monotone}): increasing when precision is low, but decreasing when precision is high, i.e., the pitfall of precision. Higher precision has three effects on the principal's payoff:
\begin{enumerate}
    \item \textit{direct effect I}: bad types who exert zero effort are more likely to be rejected;
    \item \textit{direct effect II}: types with positive effort are more likely to be approved;
    \item \textit{indirect strategic effect}: marginally more bad types are incentivized to exert effort and thus more likely to be approved.
\end{enumerate}
The strategic effect arises as higher precision increases the value of signaling.

When precision is relatively low, the marginally incentivized bad types are close to good (i.e., close to the threshold $\tilde\theta$), so the indirect strategic effect vanishes; hence, the principal benefits from higher precision. In contrast, as precision becomes sufficiently high, the approval probabilities of types with zero effort and those with positive effort converge to zero and one, respectively. Consequently, the direct effects vanish, while the indirect strategic effect operates on the extensive margin (from zero to positive effort) and becomes first-order. As a result, the principal is worse off with higher precision.


The extensive margin is key to the dominance of the strategic effect at high precision. It is generated by the non-monotone marginal benefit from effort in the approval-rejection environment---highest near the approval standard and diminishing away from it---and is robust to general costs and binary effort (see \Cref{sect:simulation}). This feature is natural to practical settings like college admissions and grant/loan applications. To show that this structural feature---rather than stronger signaling incentives per se---is the source of the pitfall, we contrast our model with a linear reputation example (\Cref{rm:extensive}), in which higher precision intensifies signaling on the intensive margin yet the pitfall never arises.

Despite the principal's pitfall of precision, the implications for agents differ. Higher precision increases the agent's expected approval rate and payoff (Proposition \autoref{prop:approval}), so agents on average prefer higher precision.

From the perspective of \textit{statistical discrimination} \`a~la \citet{phelps1972statistical}, this suggests that groups 
with more precise screening technologies are favored at the group level, 
consistent with classical statistical discrimination theory. However, the principal may actually favor groups with noisier technologies ex ante, precisely because they behave less strategically. If the principal solicits applications before observing signals, she will prefer to target groups with noisier technologies. This gives rise to a \emph{reversal of discrimination}: the group \emph{favored} at the solicitation stage is \emph{disfavored} at the approval stage. 

To clarify the informational mechanism behind the pitfall of precision, we explore the impact of precision on screening efficiency (Proposition \autoref{prop:error}) and information transmission (Proposition \autoref{prop:accuracy}). 
Higher precision increases selection accuracy between good types and sufficiently bad types, but strategic responses also make it harder to distinguish good types from moderately bad ones.

In addition to the agent's strategic responses, the principal's lack of commitment power is another source of the problem. 
If the principal can commit ex ante to an approval standard, the pitfall will disappear: higher precision will always improve her welfare (\Cref{prop:commit}). By committing to a high standard, only good types of the agent are marginally incentivized by higher precision. 
Our analysis thus highlights both the pitfall of precision in noisy screening and the institutional role of commitment in overcoming it. 

In terms of analysis, our comparative statics rely on a high-precision ($\precision\to\infty$) asymptotic analysis that characterizes the limiting behavior of equilibrium objects when closed-form solutions are unavailable. This allows us to identify which effects are first-order and which vanish (faster) in the limit, compare the magnitudes of competing forces, and establish sharp comparative statics.

\subsection*{Related Literature}
\paragraph{Impact of Noise on Information Transmission and Screening.} 
Our work relates to a literature on how noise affects information transmission and screening. In cheap talk games, \citet{blume2007noisy} show that introducing small noise can improve communication and welfare for both parties, because noise can soften conflicts of interest between them---a mechanism absent in our setting, where the agent always prefers approval. 
\cite{rick2013benefits} studies the effect of noise on information transmission in a general setting including cheap talk and signaling, but focuses on the existence of noise structures that improve welfare. Recently, \cite{tsakas2021noisy} study how noise affects sender welfare in Bayesian persuasion settings.

One closely related paper is \cite{ekmekci2022learning}, who study a reputation game where a principal gradually learns about an agent's type from a noisy performance measure that the agent can costly manipulate (like costly signaling). They similarly find that too much precision gives the bad type strong incentives to mimic the good type, inhibiting learning and harming the principal. Instead, we study the impact of noise in a static costly signaling model with continuous types.\footnote{\cite{de2011noisy} and \cite{jeitschko2012signaling} study noisy signaling with binary types, but focus only on equilibrium behavior relative to the noiseless benchmark.} Also related is \citet{adda2024grantmaking}, who study noisy screening in grantmaking with costly participation and a fixed budget, finding that more evaluation noise raises equilibrium participation.\footnote{\cite{rosar2017test} and \cite{harbaugh2018coarse} study the receiver's optimal test in environments where the sender decides whether to participate in the test. Both papers find that the optimal test features coarse grading (i.e., noisy tests) to increase participation.} In contrast, with costly signaling, we show that more evaluation noise discourages agents, particularly unqualified ones, from actively signaling. Both of us leverage \citeauthor{lehmann1988comparing}'s (\citeyear{lehmann1988comparing}) quantile-function approach to derive sharp predictions on the impact of evaluation noise.

Finally, \cite{krishna2025pareto} model college admissions as a contest among heterogeneous students and characterize when pooling test outcomes (i.e., making test performance noisier) improves all students' welfare. The Pareto improvement exploits the strategic effect of pooling on effort.\footnote{Other recent papers study noise and total effort in complete-information contests \citep[see][]{drugov2020noise,morgan2022limits}. In particular, \citet{morgan2022limits} show that total effort in large contests is single-peaked in the noise level.} In contrast, we focus on the college's screening problem without competition among students. As we discuss at the end of \Cref{sect:principal-commit}, competition can help mitigate the pitfall of precision by granting the principal some commitment power.\footnote{Indeed, pooling never benefits colleges in \cite{krishna2025pareto}, as the outcome without pooling is assortative matching, which is efficient for colleges.}



\paragraph{Strategic Manipulation.} 
This paper is also related to a recent literature on strategic manipulation. In a closely related approval model, \cite{perez2022test} study the receiver's optimal test design when the agent can costly manipulate test inputs. Falsification can be viewed as signaling: falsified inputs correspond to signaling effort, and both tests and noisy signals are information structures that stochastically map inputs to observed signals. \citeauthor{perez2022test} show that the optimal test deliberately incorporates noise (i.e., randomization) 
to deter costly falsification by low types. 
Differently, we treat information structures as exogenous and focus on comparative statics across precision levels; our ``pitfall of precision'' thus conveys a related yet distinct message: more noise can be favored even away from the optimal information structure, which matters for the statistical discrimination perspective.

Instead of exogenous noise or optimal tests, \cite{frankel2019muddled,frankel2022improving} consider a signaling environment where senders are heterogeneous in both intrinsic quality and gaming ability in manipulating signals. They show that, as stakes rise or signaling technology becomes more manipulable, more information is revealed about gaming ability, and less about intrinsic quality. Therefore, the receiver should under-utilize signals if she can commit, similar to what we find under principal commitment. Relatedly, \cite{ball2025scoring} shows in a multi-feature extension of \citeauthor{frankel2019muddled}'s (\citeyear{frankel2019muddled}) model that, compared to full disclosure, the receiver prefers a noisier scoring rule that aggregates features into a one-dimensional score, as it mitigates her commitment problem and discourages strategic manipulation.

\paragraph{Costly Signaling with Noise.} The costly signaling literature dates back to \citet{spence1973job}, and noisy signaling to \citet{matthews1983equilibrium} and \citet{carlsson1997noise}; see \cite{heinsalu2014noisy,heinsalu2018dynamic} and \cite{dilme2019dynamic} for dynamic noisy signaling. Most related are \cite{de2011noisy} and \cite{jeitschko2012signaling}, who explore how noise affects signaling behavior with binary types and experimentally test their predictions. In contrast, we study a continuous-type model, which proves more tractable and enables richer comparative statics. We also show how the vanishing noise limit serves as an equilibrium refinement for the noiseless game: it selects the pooling equilibrium and a semi-separating one, whereas standard refinements (e.g., the intuitive criterion or D1) predict a continuum of equilibria.

\paragraph{Statistical Discrimination.} 
Our paper also relates to the classical literature on statistical discrimination \citep{phelps1972statistical,arrow1973theory}; see \citet{onuchic2025recent} for a recent survey. Our model marries statistical discrimination to costly signaling: as in equilibrium statistical discrimination \citep{arrow1973theory,coate1993will}, agents' strategic choices are shaped by their ability to convey productivity, but signaling here is purely wasteful rather than a human capital investment. Our discussion of discrimination reversal connects to \citet{fryer2007belief} and \citet{bohren2019dynamics}, who show that discrimination against a group at one stage of an evaluation process can reverse at a later stage; our reversal, however, arises through the pitfall of precision rather than dynamic belief updating.

\section{A Noisy Signaling Model}
\label{sect:eqm}
\subsection{Model Setup}

A principal (she) decides whether to approve an application by an agent (he). The agent has a private type $\theta\in\Theta:=[\underline\theta,\overline\theta]\subset\mathbb{R}_{++}$. The principal does not know $\theta$ but believes that $\theta$ is distributed according to a cumulative distribution function (cdf) $G$, which admits a probability density function (pdf) $g$ such that $g(\theta)>0$ 
for all $\theta\in[\underline\theta,\overline\theta]$. If the application is approved, the agent gets a payoff of $1$, independent of his type, and the principal gets $v(\theta)$, where $v$ is continuous, bounded and weakly increasing; otherwise, both players get zero payoff. For technical convenience, assume that $v$ and $g$ are continuously differentiable.

Let $\tilde\theta:=\inf\{\theta\in[\underline\theta,\overline\theta]:v(\theta)\geq0\}$; hence $\theta<\tilde\theta$ are bad types and $\theta\geq\tilde\theta$ are good types. To make the problem non-trivial, assume that $\mathbb{E}[v(\theta)|\theta\geq\tilde\theta]>0$ where the expectation is taken with respect to $G$. For expositional convenience, we will mainly focus on the case with $\mathbb{E}[v(\theta)]\leq0$, i.e., when the principal is a priori pessimistic about the agent's type. Therefore, absent any extra information, the principal would not approve the application.\footnote{This case can be interpreted as a ``cherry-picking'' market, as in \cite{bartovs2016attention}, where only top applicants are selected from a large pool of candidates.

The case with $\mathbb{E}[v(\theta)]>0$ is included in \Cref{sect:optimistic} for completeness. The pitfall of precision holds true as long as $\mathbb{E}[v(\theta)/\theta]<0$, where $1/\theta$ is type $\theta$'s marginal effort cost. Instead, when $\mathbb{E}[v(\theta)/\theta]\geq0$, regardless of the noise level, only a pooling equilibrium exists in which all types pool on zero effort and the principal approves every type, which is not of our interest.} 

\begin{assumption}[Pessimistic Prior]
\label{as:ex-ante-bad}
$\mathbb{E}[v(\theta)]\leq0$.
\end{assumption}

\paragraph{Noisy Signaling/Screening.} Before the principal makes her decision, the agent who privately knows his own type can make some signaling effort $e\in E:=[0,\infty)$. For each type, effort cost is \textit{linear}, given by $C(e,\theta)=e/\theta$; it is equivalent to assuming a general constant marginal cost $c(\theta)$ which is continuous and strictly decreasing in $\theta$.\footnote{We assume linear costs for simplicity. Our main result can extend to more general costs; see the discussion and numerical simulations in 
\Cref{sect:simulation}.} Instead of directly observing the effort level $e$ as in the standard signaling model, the principal can only observe a \textit{noisy signal} $s$ about $e$, given by $s=e+\epsilon/\precision$, where $\epsilon$ denotes a base noise which has a cdf $F$ and a pdf $f$ such that $\mathbb{E}[\epsilon]=0$, and $\precision>0$ captures the precision level which is our main focus. 
Let $f_{\precision}(z):=\precision f(\precision z)$ for $z\in\mathbb{R}$ denote the pdf of $\epsilon/\precision$.

\begin{assumption}
\label{as:noise}
The pdf of the noise $f$ satisfies the following assumptions:
\begin{enumerate}
    \item unbounded support: $f(z)>0$ for all $z\in\mathbb{R}$;\footnote{
    Our results can extend to bounded noise; see \Cref{sect:bounded-noise}.
    }
    \item log-concavity: $\log f(z)$ is concave in $z\in\mathbb{R}$;
    \item strictly diminishing bad luck: $f(z)<f(z')$ for any $z<z'<0$ such that $f(z)>0$.
\end{enumerate}
\end{assumption}
Given (1) and (2) in \Cref{as:noise}, the signal structure from $E$ to $\Delta(S)$, denoted by $h_{\precision}(s|e):=f_{\precision}(s-e)$, has full support and satisfies the Monotone Likelihood Ratio Property (MLRP). Moreover, $h_{\precision}(s|e)$ is ranked with respect to the accuracy order \citep{lehmann1988comparing,persico2000information}: the larger $\precision$ is, the more accurate the signal is. Hence, $\precision$ measures the signaling/screening precision.

To ease the exposition, further assume that $f$ is symmetric, continuous, and twice continuously differentiable on its support except for $0$. Define $f^{-1}(p):=\inf\{z\geq0:f(z)=p\}$ and $f^{-1}_\precision(p):=\inf\{z\geq0:f_{\precision}(z)=p\}$, hence $f_{\precision}^{-1}(p)=\frac{1}{\precision}f^{-1}(p/\precision)$. By all the assumptions, $f^{-1}:[0,f(0)]\to\overline{\mathbb{R}}_+$ is a continuous, decreasing function, with the convention that $f^{-1}(0):=+\infty$. 


Examples satisfying \Cref{as:noise} and others include: normal noise, where $f(z)=\frac{1}{\sqrt{2\pi}}\text{exp}(-z^2/2)$, and exponential noise, where $f(z)=\frac12\text{exp}(-|z|)$.


\paragraph{Solution Concept.} We consider pure-strategy Perfect Bayesian Equilibrium (PBE) (henceforth, equilibrium) where each type of the agent chooses an effort level, $e:\Theta\to E=[0,\infty)$, and for each signal realization, the principal forms beliefs over types and chooses whether to approve. Since the signal structure has full support, there is no off-path signal, therefore we do not need to explicitly specify the principal's belief system, let alone imposing refinements.

\subsubsection{Discussion of the Model} 
\label{sect:model discussion}
\paragraph{Signaling vs. Screening Noise.} The noise in our model can arise from either the signaling process or the principal's evaluation process. In the former case, signals such as test scores or résumés are noisy reflections of ability. In the latter case, noise stems from how interviews are conducted or from the principal's inattention or cognitive limitations in evaluating signals. 

\paragraph{Signaling Technology: Complements vs. Substitutes.} Following \cite{spence1973job}, we assume that signaling costs are negatively correlated with agent types (e.g., productivity) and that signals (e.g., education or test scores) reflect types through potentially different signaling effort. This setup is isomorphic to one in which signaling costs are identical across types, while types and effort jointly determine the agent's signal in a complementary way: $s=\theta\cdot e+\epsilon$ and $C(e)=e$; by redefining $\hat{e}:=\theta \cdot e$, it goes back to our original setup. We adopt this setup for simplicity: the complementarity ensures that the agent's strategy is monotone (where higher types always exert more effort), which implies a simple cutoff structure for the agent's equilibrium strategy, as we demonstrate later. 

Suppose instead that types and effort are substitutes; e.g., $s=\theta+e+\epsilon$ with $C(e)=e$.\footnote{Note that this is isomorphic to the setup in which $s=e+\epsilon$ with $C(e,\theta)=|e-\theta|$, which can be interpreted as a model with costly falsification \citep{perez2022test}.} 
Then signaling incentives are non-monotone (though higher types still generate higher signals): in equilibrium, both low types and very high types exert zero effort, as very high types can substitute their innate ability for effort. This non-monotonicity complicates equilibrium analysis without providing new insights on the impact of precision. We believe the direct and indirect effects described in the introduction---which we explore in detail later---still apply to this alternative signaling technology: in particular, an increased precision will still incentivize more bad types to invest in signaling, which is harmful to the principal.

\paragraph{A Population Interpretation: Without vs. With Quotas.} Our model also admits a population interpretation: there is a continuum of agents of mass 1, and $G$ is the population distribution of types. The principal must decide whether to approve each agent's application, but cannot distinguish among agents or types, so she can only base her decisions on noisy signals.

In the baseline model, the principal faces no binding quota constraint and may approve as many agents as she wishes. We believe this is appropriate for many practical settings, such as college admissions, grantmaking, and recruitment, where quotas are soft and approval decisions are relatively flexible. We briefly discuss at the end of \Cref{sect:principal-commit} what happens when the principal instead faces a strict quota requiring her to approve a fixed number of agents, a constraint that effectively endows the principal with commitment power.


\subsection{Equilibrium Characterization} 
In this section, we characterize equilibria for a given precision level.

\paragraph{Approval Standard.} Since higher types of the agent face lower effort costs, they must invest weakly more than lower types in equilibrium. Higher effort then generates higher signals under the signal structure with the Monotone Likelihood Ratio Property. Therefore, upon observing a high signal, the principal should believe that it more likely comes from a high type. In equilibrium, the principal thus takes a cutoff strategy: approving if and only if the observed signal $s$ exceeds some cutoff $\tau\in\overline{\mathbb{R}}:=\mathbb{R}\cup\{\pm\infty\}$. Hence, we can think the principal's problem as simply choosing a cutoff $\tau$ in the signal space, which is interpreted as the approval standard of the screening process.

\paragraph{Equilibrium Characterization.}
As in standard signaling games, a pooling equilibrium always exists where all agent types exert zero effort and the principal sets $\tau=+\infty$ and never approves. In this pooling equilibrium, signals are uninformative so that the principal follows her ex ante optimal choice regardless of signal realizations, in which case no agent type has incentives to exert effort. This is the only equilibrium when the precision is very low.

In contrast, when the precision is not too low, we show that at least a semi-separating equilibrium also exists. In this equilibrium, the principal sets an interior standard $\tau_-$ and high types above some threshold $\hat\theta_-\in[\underline\theta,\tilde\theta]$ are incentivized to exert positive effort, including both good types (i.e., above $\tilde\theta$) and some bad types (i.e., between $\hat\theta_-$ and $\tilde\theta$). 

\begin{proposition}
\label{prop:eqm}
    The following hold:
    \begin{itemize}
        \item[(a)] A pooling equilibrium always exists where the principal rejects upon all signals 
        and $e(\theta)\equiv0$ for all $\theta\in[\underline\theta,\overline\theta]$. It is the unique equilibrium when $\precision$ is sufficiently small.
        \item[(b)] There exists $\tilde{\precision}>0$ such that, when 
        $\precision\in[\tilde{\precision},\infty)$, a semi-separating 
        equilibrium also exists: the principal sets an approval standard 
        $\tau_-(\precision)\in(0,+\infty)$, and agent types above some 
        threshold $\hat\theta_-(\precision)\in(\underline\theta,\tilde\theta]$ 
        exert strictly positive effort while those below exert zero effort.
    \end{itemize}
\end{proposition}
The threshold $\tilde{\precision}$ is formally defined in \Cref{eq:tildesigma} in \Cref{app}.

Other equilibria with interior standards may also exist. However, as we show at the end of this section, only the two equilibria described in \Cref{prop:eqm} survive as noise diminishes; moreover, the semi-separating equilibrium with $(\tau_-,\hat\theta_-)$ is the most efficient equilibrium when precision is high.

Now we sketch the idea behind this equilibrium characterization and relegate the formal analysis and proof to \Cref{app}. To characterize equilibria, especially the semi-separating one in \Cref{prop:eqm}, we study the agent's best response to a given standard and the principal's optimal approval standard given the agent's strategy. Equilibria are then in mutual best responses.

\paragraph{Step 1: Agent's Best Response in Threshold $\hat\theta$.}
Fix an approval standard $\tau$. Notice that the marginal signaling benefit $\frac{\mathrm{d}}{\mathrm{d}e}\text{Pr}(e+\epsilon/\precision>\tau)=f_{\precision}(\tau-e)$ is single-peaked around the standard $\tau$. Therefore, the agent can either ``give up'' and exert zero effort, in which case he may get approved purely by luck. Or, he can exert positive effort to exceed the standard but may still probably get rejected due to bad luck. The agent needs to compare which is better to decide his best response (as  his problem is not concave; see \Cref{app:eqm-proof} for details).

Conditional on exerting positive effort, the optimal level is characterized by the first-order condition where the marginal effort cost equalizes the marginal benefit of avoiding bad luck:
\[
1/\theta=\frac{\mathrm{d}\text{Pr}(e+\epsilon/\precision>\tau)}{\mathrm{d}e}=\frac{\mathrm{d}[1-F_\precision(\tau-e)]}{\mathrm{d}e}=f_{\precision}(\tau-e).
\]
Let $\hat{e}(\theta;\tau)\geq\tau$ be the optimal positive effort level satisfying this condition.\footnote{To be precise, $\hat{e}(\theta;\tau)$ is only well-defined when $1/\theta\leq f_{\precision}(0)=\precision f(0)$, that is, when $\precision$ is not too small; otherwise, let $\hat{e}(\theta;\tau):=\tau$. 
See \Cref{app:eqm-proof} for the formal treatment.} 

The agent of type $\theta$ compares $e=0$ with $e=\hat{e}(\theta;\tau)$. It can be verified that a threshold type $\hat\theta(\tau;\precision)\in[\underline\theta,\overline{\theta}]$ exists such that all types above it strictly prefer $e=\hat{e}(\theta;\tau)$ and types below strictly prefer $e=0$. As a result, we can think that the agent is just choosing a threshold $\hat\theta$ in the type space in response to $\tau$. 

The threshold type $\hat\theta$ is determined by the indifference condition between $e=0$ and $e=\hat{e}(\theta;\tau)$ when $\hat\theta$ is interior, i.e., $\hat\theta\in(\underline\theta,\overline\theta)$:
\begin{equation}
\label{eq:S-BR}
    1-F_\precision(\tau-0)=1-F_\precision(\tau-\hat{e}(\hat\theta;\tau))-\hat{e}(\hat\theta;\tau)/\hat\theta. \tag{A-BR}
\end{equation}
Let $\widehat\Theta(\tau)$ denote the solution to \Cref{eq:S-BR} (which is well-defined; see \Cref{lemma:S-BR-increasing} in \Cref{app:eqm-proof}). It captures the agent's best response.
\begin{lemma}
\label{lemma:Theta-increasing}
    When $\tau\geq0$, $\widehat\Theta(\tau)$ is increasing in $\tau$.
\end{lemma}
The intuition is simple: as the standard increases, more and more low types give up trying to meet the standard. The part with $\tau\leq0$ is less important, as it turns out that the principal never sets $\tau\leq0$ in equilibrium.

\paragraph{Step 2: Principal's Best Response.}

Given the agent's increasing effort strategy $e(\theta)$, if the principal's optimal approval standard $\tau$ is interior, she should be indifferent between approving and rejecting upon observing a signal $s=\tau$. That is, (an interior) $\tau$ should solve
\[
\mathbb{E}[v(\theta)|e(\theta)+\epsilon/\precision=\tau]=0\quad\Leftrightarrow\quad
\int_{\underline\theta}^{\overline\theta} v(\theta)g(\theta)f_{\precision}(\tau-e(\theta))\mathrm{d}\theta=0.
\]

We have derived in Step 1 that in equilibrium, for some threshold $\hat\theta$, $e(\theta)=0$ for $\theta<\hat\theta$ and $e(\theta)=\hat{e}(\theta;\tau)$ for $\theta\geq\hat\theta$; in particular, $f_{\precision}(\tau-\hat{e}(\theta;\tau))=1/\theta$. Plugging them into the indifference condition, the equilibrium condition for the principal's standard $\tau$ is thus:
\begin{equation}
\label{eq:R-BR}
f_{\precision}(\tau-0)\int_{\underline\theta}^{\hat\theta} v(\theta)g(\theta)\mathrm{d}\theta+\int_{\hat\theta}^{\overline\theta} v(\theta)g(\theta)/\theta\mathrm{d}\theta=0.\tag{P-BR}
\end{equation}

Since $\mathbb{E}[v(\theta)]\leq0$, the first term in the left-hand side of \Cref{eq:R-BR} is always negative, hence in any equilibrium involving an interior $\tau$, we must have $\int_{\hat\theta}^{\overline\theta} v(\theta)g(\theta)/\theta\mathrm{d}\theta\geq0$. That is, the agent's equilibrium threshold $\hat\theta$ cannot be too low: it must be $\hat\theta\geq\theta^\dagger$, where
\begin{equation}
\label{eq:thetadagger}
\theta^\dagger:=\inf\bigg\{\hat\theta\in[\underline\theta,\overline\theta]:\int_{\hat\theta}^{\overline\theta}v(\theta)g(\theta)/\theta\mathrm{d}\theta\geq0\bigg\}<\tilde\theta.
\end{equation}

The first term in the left-hand side of \Cref{eq:R-BR} captures the marginal cost from lowering the approval standard, and the second term is the marginal benefit. With a lower standard, the principal can approve more high types with positive effort at a size of $g(\theta)/\theta$, but also approve more low types with zero effort marginally with a size of $g(\theta)f_{\precision}(\tau-0)$. The benefit and the cost should be balanced at the optimum, reflected by \Cref{eq:R-BR}. 

Let $T(\hat\theta)$ be the positive standard $\tau\geq0$ that solves \Cref{eq:R-BR}; by symmetry of $f$, $-T(\hat\theta)$ solves \Cref{eq:R-BR} too. Intuitively, $\pm T$ capture the principal's ``best response'' to the agent's threshold $\hat\theta$.\footnote{To be precise, $\pm T$ are not exactly the principal's best responses, but rather equilibrium conditions on $\tau$, since in deriving \Cref{eq:R-BR} we already use the agent's best response.}

\begin{lemma}
\label{lemma:T-Ushaped}
    $T(\hat\theta;\precision)$ is strictly decreasing for $\hat\theta\in(\theta^\dagger,\tilde\theta)$ and increasing for $\hat\theta\in(\tilde\theta,\overline\theta)$, with $T(\theta^\dagger;\precision)=+\infty$.
\end{lemma}
That is, the principal's optimal approval standard $T$ is U-shaped in $\hat\theta$. Recall that $\tilde\theta$ is the threshold separating good and bad types. Intuitively, when $\hat\theta<\tilde\theta$, since a higher $\hat\theta$ means fewer bad types exert positive effort, the principal would like to lower the standard to approve more good types (without having to approve too many bad types since more of them pool on zero effort). In contrast, when $\hat\theta>\tilde\theta$, as $\hat\theta$ increases, fewer good types exert positive effort, which implies a smaller marginal benefit of lowering the standard, and thus in response the principal would set a higher standard.

\paragraph{Step 3: Equilibrium via Mutual Best Responses.} Equilibria are characterized by fixed points of the two best responses, $\widehat\Theta\circ T$ or $\widehat\Theta\circ(-T)$. It can be shown that $\widehat\Theta\circ(-T)$ has no fixed point. And for $\widehat\Theta\circ T$, given that $T$ is U-shaped in $\hat\theta$ and $\widehat\Theta$ is increasing in $\tau$, on the decreasing part of $T$, i.e., to the left of $\tilde\theta$, $\widehat\Theta\circ T$ must have a unique fixed point, denoted by $(\tau_-,\hat\theta_-)$, so long as the graphs of $T$ and $\widehat\Theta$ ever intersect, which happens when $\precision$ is larger than a threshold $\tilde{\precision}>0$. There are possibly other fixed points on the increasing part of $T$, including the pooling equilibrium $(\tau,\hat\theta)=(+\infty,\overline\theta)$. \Cref{prop:eqm} thus follows.

\paragraph{Vanishing Noise and Equilibrium Selection.}
Let $(\tau_-(\precision),\hat\theta_-(\precision))$ denote the semi-separating equilibrium identified in \Cref{prop:eqm}, indexed by the precision level $\precision$. It can be shown that $\lim_{\precision\to\infty}(\tau_-(\precision),\hat\theta_-(\precision))=(\theta^\dagger,\theta^\dagger)$, which characterizes the vanishing noise limit of the semi-separating equilibria, where $\theta^\dagger$ is defined in \Cref{eq:thetadagger}. The following result shows that only this limit equilibrium and the pooling one survive in the vanishing noise limit.

\begin{proposition}
\label{prop:limit}
    In the limit of $\precision\to\infty$, equilibria converge to either the pooling equilibrium or the limit semi-separating one: for any sequence of $\precision_k>0$ and equilibria under $\precision_k$ with $\hat\theta_k$ such that $\lim_{k\to\infty}\precision_k=\infty$, it holds that $\lim_{k\to\infty}\hat\theta_k=\overline\theta$ or $\theta^\dagger$.
\end{proposition}

As a result, as $\precision\to\infty$, there are only two scenarios: either all types above $\theta^\dagger$ pool on effort level $e=\theta^\dagger$, while types below exert zero effort, and the principal only approves applications from types with positive effort $e=\theta^\dagger$; or all types pool on zero effort and the principal does not approve any application. 

The vanishing noise limit of this noisy signaling model thus provides a refinement for equilibria in the noiseless model. Without noise, in addition to the pooling equilibrium, there is a continuum of ``semi-separating'' weak PBE, indexed by thresholds $\hat\theta$ such that $\mathbb{E}[v(\theta)|\theta\geq\hat\theta]\geq0$, where types above $\hat\theta$ pool on effort level $\hat\theta$, types below exert zero effort, and the principal only approves types with effort level $\hat\theta$. Among these equilibria, all semi-separating ones with threshold $\hat\theta\leq\tilde\theta$ can pass the D1 criterion \citep{cho1987signaling,banks1987equilibrium}. In contrast, our vanishing noise refinement selects a unique semi-separating equilibrium with $\hat\theta=\theta^\dagger$, alongside the pooling one.

As a corollary of \Cref{prop:limit}, when $\precision$ is large, both the principal's and the agents' payoffs in equilibrium $(\tau_-(\precision),\hat\theta_-(\precision))$---approximately $\int_{\theta^\dagger}^{\overline\theta}v(\theta)g(\theta)\mathrm{d}\theta>0$ for the principal and $\max\{1-\theta^\dagger/\theta,0\}$ for the agent of type $\theta$---are higher than the zero or almost-zero payoffs in any other equilibrium. 

\begin{corollary}
\label{cor:pareto}
When $\precision$ is sufficiently large, the semi-separating equilibrium identified in \Cref{prop:eqm} is Pareto-optimal among all equilibria.
\end{corollary}
For comparative statics, we thus focus on this semi-separating equilibrium.

\section{The Pitfall of Precision}
\label{sect:cs}
This section studies the impact of the precision in the signaling process on equilibrium outcomes. Among other things, we highlight a strategic effect of precision: a higher precision increases the value of signaling and incentivizes more low types to exert effort. In \Cref{sect:main result}, we show that, due to this strategic effect, the principal's welfare is non-monotone in the precision. In particular, we identify a pitfall of precision: the principal is worse off when the process becomes too precise (i.e., $\precision$ large). In \Cref{sect:accuracy}, we further unravel the pitfall of precision by discussing the implications of precision on the endogenous screening efficiency and accuracy of information transmission.

\subsection{Welfare Implications of Precision}
\label{sect:main result}
To understand how the precision affects the principal's payoff, let us first investigate its impact on equilibrium behavior. Recall that the equilibrium of interest is the semi-separating one $(\tau_-(\precision),\hat{\theta}_-(\precision))$ in \Cref{prop:eqm} when $\precision\in[\tilde{\precision},\infty)$. 

Let $\hat\tau_-(\precision):=\precision\tau_-(\precision)$ be the \emph{``precision-adjusted'' standard}. We are interested in it because $\hat\tau$ is a more relevant policy variable for both the principal and the agent by taking the noise into account; for example, it directly characterizes the approval probability at zero effort: $\text{Pr}(0+\epsilon/\precision>\tau)=1-F(\precision\tau)=1-F(\hat\tau)$. 
\begin{lemma}
\label{lemma:increasing-indifference}
    For $\precision\in[\tilde{\precision},\infty)$, $\hat\theta_-(\precision)$ is strictly decreasing 
    and $\hat\tau_-(\precision)$ is strictly increasing. 
\end{lemma}
In words, as the precision increases, more types engage in costly signaling (as the equilibrium threshold type is decreasing in precision). This is because greater precision increases the value of signaling: with greater precision, an agent who pays the cost to exceed the standard can be more confident that doing so will lead to approval.

On the other hand, the equilibrium precision-adjusted standard increases with precision. In a less noisy environment, the principal can afford to be more lenient by lowering the standard, while still limiting the approval probability of low types with zero effort, $1-F(\precision\tau)$. This leniency allows her to approve more high types with positive effort. However, to balance this trade-off, the principal does not lower the standard proportionally with precision, leading to an increased precision-adjusted standard.

Now we are ready to explore the welfare implications of the precision. 

\paragraph{Principal Payoff.}
Let $V(\precision)$ be the principal's ex ante payoff at $\precision$. Recall that $\text{Pr}(0+\epsilon/\precision>\tau)=1-F(\hat\tau)$, and given that $1/\theta=f_{\precision}(\tau-\hat{e}(\theta))=\precision f(\precision(\tau-\hat{e}(\theta)))$ and $\hat{e}(\theta)\geq\tau$, it can be shown that 
    \[
    \text{Pr}\left(\hat{e}(\theta;\tau)+\epsilon/\precision>\tau\right)=1-F\left(\precision(\tau-\hat{e}(\theta))\right)=1-F(-f^{-1}(1/(\precision\theta))).
    \]
Therefore, the principal's ex ante payoff can be written as
\begin{align*}
V(\precision)
=\int_{\underline\theta}^{\hat\theta_-(\precision)} v(\theta)g(\theta)\mathrm{d}\theta\cdot [1-F(\hat\tau_-(\precision))]+\int_{\hat\theta_-(\precision)}^{\overline\theta} v(\theta)g(\theta)[1-F(-f^{-1}(1/(\precision\theta)))]\mathrm{d}\theta.
\end{align*}
An increase in precision has three effects:
\begin{enumerate}
    \item \textbf{Rejecting more low types}: By \Cref{lemma:increasing-indifference}, $1-F(\hat\tau_-(\precision))$ is decreasing in $\precision$. Therefore, as the precision increases, bad types who exert zero effort are less likely to be approved simply by good luck. 
    This direct effect leads to a decrease of the Type II error and is good for the principal.
    
    \item \textbf{Approving more high types}: An increased precision also limits the scope of bad luck and leads to more approval of high types ($\theta>\hat\theta_-(\precision)$) who actively invest effort. 
    In detail, 
    \[
    \frac{\mathrm{d}}{\mathrm{d}\precision}[1-F(-f^{-1}(1/(\precision\theta)))]=-\frac{1}{\precision^3\theta^2 f'(f^{-1}(1/(\precision\theta)))}>0.
    \]
    Approving more good types decreases the Type I error and is surely beneficial, but not all types above $\hat\theta_-(\precision)$ are good types. A higher precision thus also results in bad types in $(\hat\theta_-,\tilde\theta)$ being more likely to be approved, generating loss. Hence, the overall effect is ambiguous for the principal.
    
    When $\precision$ is close to $\tilde{\precision}$, the equilibrium threshold $\hat\theta_-(\precision)$ is close to $\tilde\theta$, that is, very few bad types invest effort, so the loss is negligible compared to the benefit and thus this second direct effect is also good for the principal. 
    
    In contrast, when $\precision$ is large, $\hat\theta_-(\precision)$ is close to $\theta^\dagger$, the loss from falsely approving more bad types outweighs the benefit from rejecting fewer good types. This is due to the diminishing bad luck and the log-concavity of $f$, which imply that the marginal increase in the approval probability is smaller for higher types who exert higher effort.  
    
    \item \textbf{Strategic effect}: More importantly, the precision also affects the agent's strategic signaling behavior. According to \Cref{lemma:increasing-indifference}, as the precision increases, more low types are incentivized by the increased precision to invest effort and try to meet the standard (in terms of a decrease in $\hat\theta_-(\precision)$). 
    Given that $\hat\theta_-(\precision)\leq\tilde\theta$ and $v(\hat\theta_-(\precision))\leq0$, the types who are incentivized are bad types, implying an increase of the Type II error. Hence, this effect generates a loss for the principal. When $\precision$ is close to $\tilde{\precision}$, the marginally incentivized types are close to $\tilde\theta$, so this negative effect vanishes at $\precision=\tilde{\precision}$.
    
\end{enumerate}

Combining these three effects, when $\precision$ is relatively low (close to $\tilde{\precision}$), the strategic effect vanishes, 
while the two direct effects of an increased precision are both good for the principal, 
therefore the principal benefits from a higher precision level.

In contrast, when precision is large ($\precision\to\infty$), the negative strategic effect becomes first-order relative to the direct effects. This is because the approval probabilities of types with zero and positive effort converge to zero and one, respectively, rendering the direct effects negligible. Meanwhile, the strategic effect operates on the extensive margin: low types switch from near-zero to near-one approval probability. In the proof, we show that the direct effects vanish faster than the strategic effect captured by $\mathrm{d}\hat\theta_-/\mathrm{d}\sigma$ (see \Cref{lm:dtheta/dsigma-order}, Claims \autoref{claim:a1} and \autoref{claim:A2}). Due to the dominant strategic effect, the principal is worse off under higher precision---a phenomenon we term the \textit{pitfall of precision}.

\Cref{prop:non-monotone} summarizes these observations.

\begin{proposition}
\label{prop:non-monotone}
    The principal's ex ante payoff $V(\precision)$ is non-monotone in $\precision$:
    \begin{itemize}
        \item[(a)] increasing when $\precision$ is low: $V'(\tilde{\precision}+)>0$, where $\tilde{\precision}$ is the threshold introduced in \Cref{prop:eqm};
        \item[(b)] decreasing when $\precision$ is high: $V'(\precision)<0$, for $\precision>\bar{\precision}$ with some $\bar{\precision}>\tilde{\precision}$.
    \end{itemize}
\end{proposition}
See \Cref{fig:V-non-mono} for illustrations with normal and exponential noise, respectively. Therefore, when the precision is high, the principal benefits from introducing some noise into the screening process.

\begin{figure}[htbp]
    \centering
    \begin{subfigure}{0.48\textwidth}
        \centering
        \includegraphics[width=\linewidth]{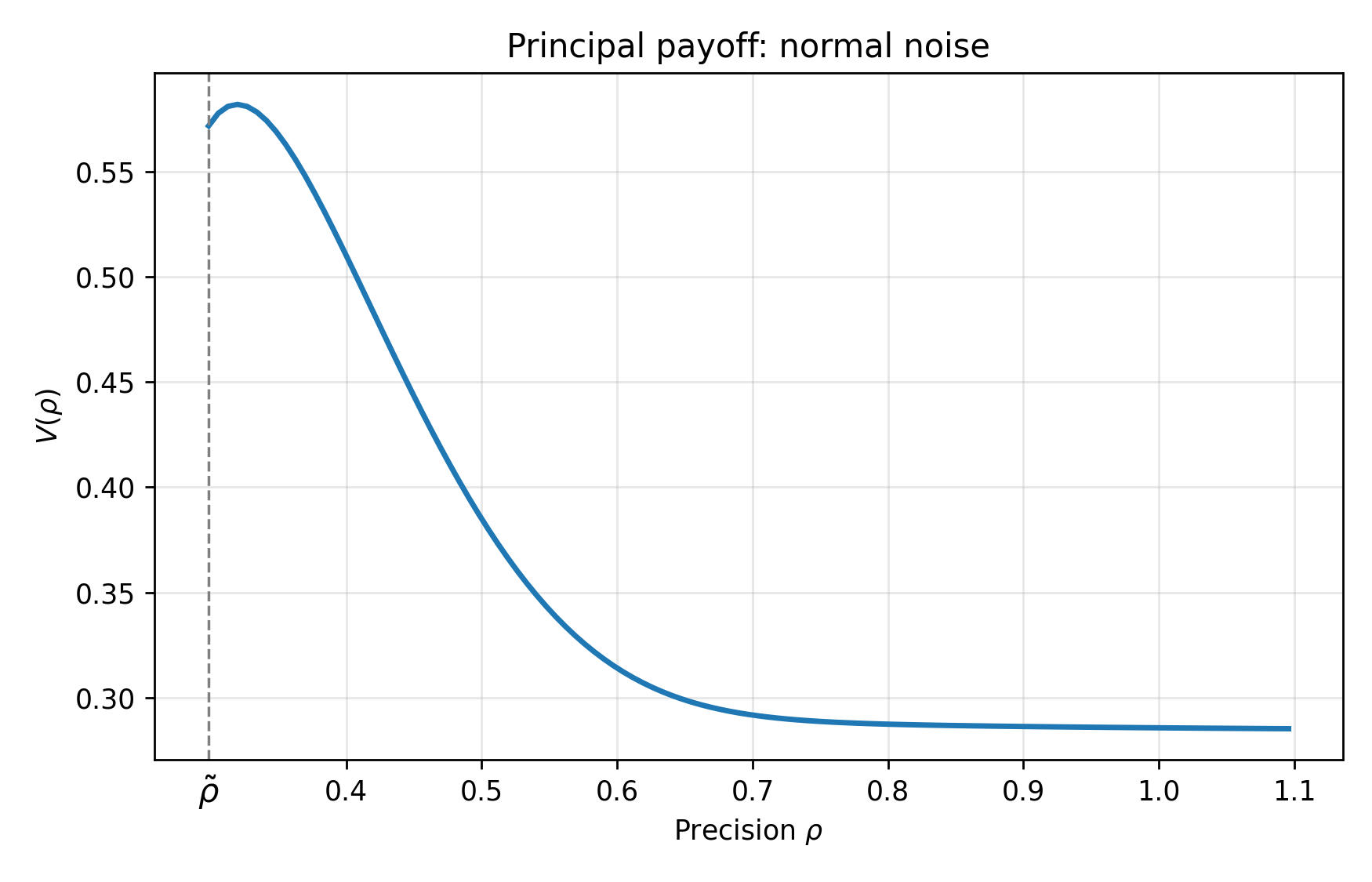}
        \caption{Normal noise.}
    \end{subfigure}
    \hfill
    \begin{subfigure}{0.48\textwidth}
        \centering
        \includegraphics[width=\linewidth]{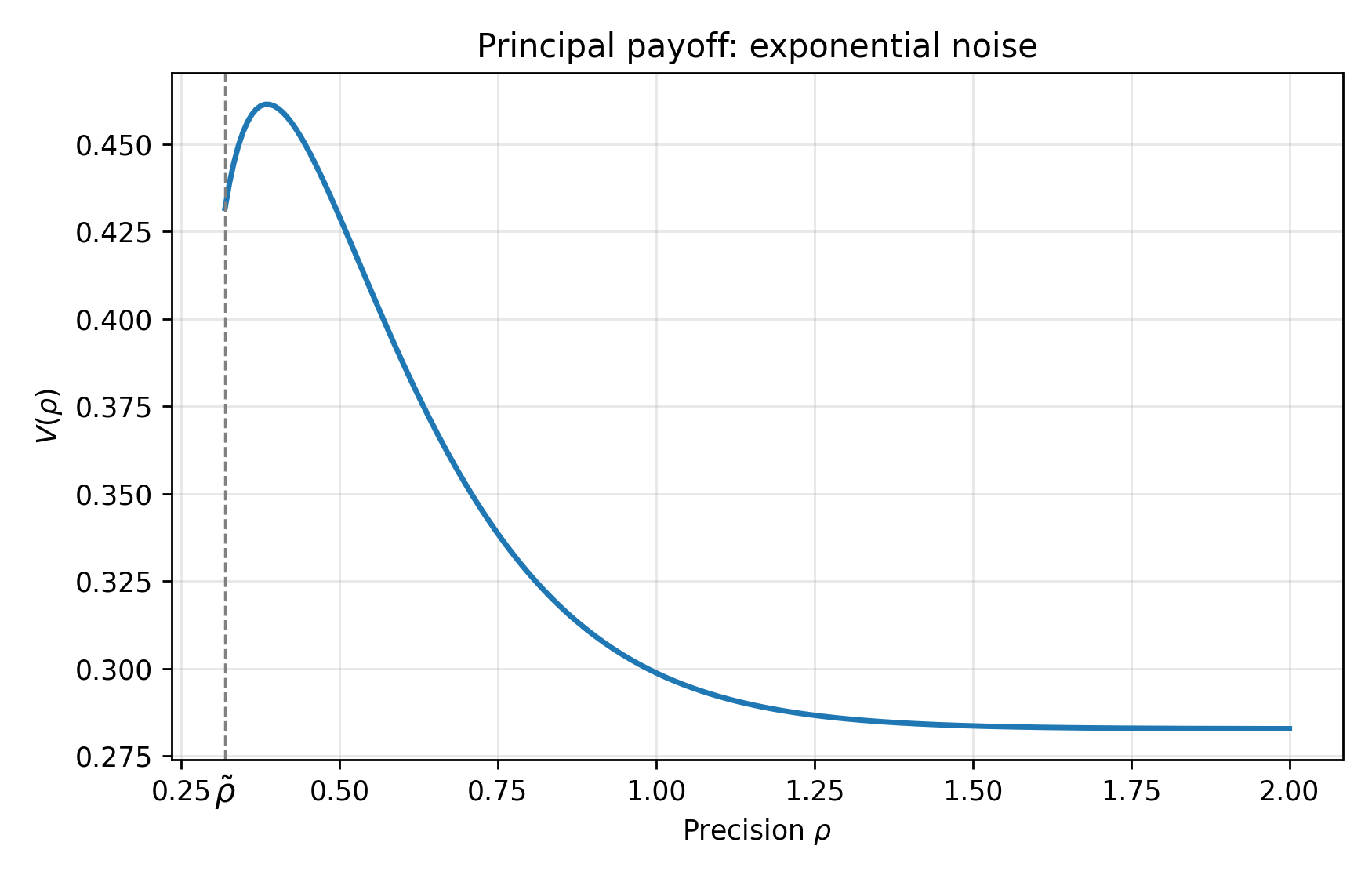}
        \caption{Exponential noise.}
    \end{subfigure}
    \caption{Non-monotonicity of principal payoff $V(\precision)$, where $\theta$ is taken to be uniformly distributed over $[5,15]$, with $v(\theta)=\theta-11$.}
    \label{fig:V-non-mono}
\end{figure}

\begin{remark}[Extensive vs. intensive margin]
\label{rm:extensive}
The pitfall of precision depends critically on the extensive margin of the strategic effect in this approval model. To see this, consider the following linear reputation example. Let
\[
    \theta \sim N(\mu,\omega^2),
    \qquad
    s=e+\epsilon/\precision,
    \qquad
    \epsilon\sim N(0,1).
\]
Observing signal $s$, the market forms a posterior mean $\mathbb{E}[\theta\mid s]$, capturing the agent's reputation. Suppose that the agent's payoff is
\[
    \mathbb{E}[\theta\mid s]-\frac{1}{2}(e-\theta)^2.
\]
The unique linear equilibrium features ``effort'' $e_{\precision}(\theta)=\theta+k_{\precision}$ and posterior mean
\[
    \mathbb{E}[\theta\mid s]
    =\mu+k_{\precision}(\theta-\mu+\epsilon/\precision)
    \sim N\!(\mu,\,q_{\precision}),
    \qquad
    k_{\precision}:=\frac{\omega^2}{\omega^2+1/\precision^2},\; q_{\precision}:=\frac{\omega^4}{\omega^2+1/\precision^2}.
\]

Higher precision still intensifies signaling incentives, but on the intensive margin across all types: as $k_{\precision}$ increases in $\precision$, every type exerts higher ``effort''. Nevertheless, the posterior mean becomes more dispersed, as $q_{\precision}$ increases in $\precision$, making equilibrium signaling more informative. Hence, if the principal has a convex payoff over posterior means, she always benefits from higher precision---no pitfall arises. 

This contrast clarifies the source of the pitfall. In our model, the approval-rejection environment generates non-monotone marginal benefit from effort (single-peaked around the standard; unlike the constant one under linear reputation), producing a zero-versus-positive effort cutoff.\footnote{
This is also by linear costs. With general convex costs, the single-peaked marginal benefit can still generate a low-versus-high effort cutoff and thus the extensive margin; see \Cref{sect:simulation}.} This creates an extensive margin for the strategic effect, which is absent in the linear reputation example. This example shows that precision-induced signaling on the intensive margin alone may never outweigh the direct benefits of precision, and thus need not generate the pitfall.

More broadly, the extensive margin is not unique to binary treatment environments. We conjecture that discrete treatments can in general generate extensive margins and hence the pitfall. Extensive margins can also arise mechanically from discrete effort choices (e.g., binary effort); see \Cref{sect:simulation} for an illustration.
\end{remark}

\paragraph{Approval Rate and Agent Payoff.} 

Let us now consider the agent's welfare. It can be shown that when the precision is high, both the equilibrium approval rate in the population and the agent's average/ex ante payoff is higher with a higher precision. 

Let $AR(\precision)$ and $U(\precision)$ be the equilibrium approval rate and agent ex ante payoff, respectively:
\begin{align*}
    AR(\precision):=&G(\hat\theta_-(\precision))[1-F(\hat\tau_-(\precision))]+\int_{\hat\theta_-(\precision)}^{\overline\theta} g(\theta)[1-F(-f^{-1}(1/(\precision\theta)))]\mathrm{d}\theta,\\
    U(\precision):=&G(\hat\theta_-(\precision))[1-F(\hat\tau_-(\precision))]+\int_{\hat\theta_-(\precision)}^{\overline\theta} g(\theta)\left\{1-F(-f^{-1}(1/(\precision\theta)))-\frac{1}{\precision\theta}[\hat\tau_-(\precision)+f^{-1}(1/(\precision\theta))]\right\}\mathrm{d}\theta.
\end{align*}

\begin{proposition}
\label{prop:approval}
    The following hold:
    \begin{itemize}
        \item[(a)] The approval rate $AR(\precision)$ is increasing when $\precision$ is large.
        \item[(b)] Suppose that $\lim_{z\to\infty}\frac{f(z)f''(z)}{[f'(z)]^2}=1$. Then the agent's ex ante payoff $U(\precision)$ is increasing when $\precision$ is large.
    \end{itemize}
\end{proposition}
\begin{figure}[!htbp]
    \centering
    \begin{subfigure}{0.48\textwidth}
        \centering
        \includegraphics[width=\linewidth]{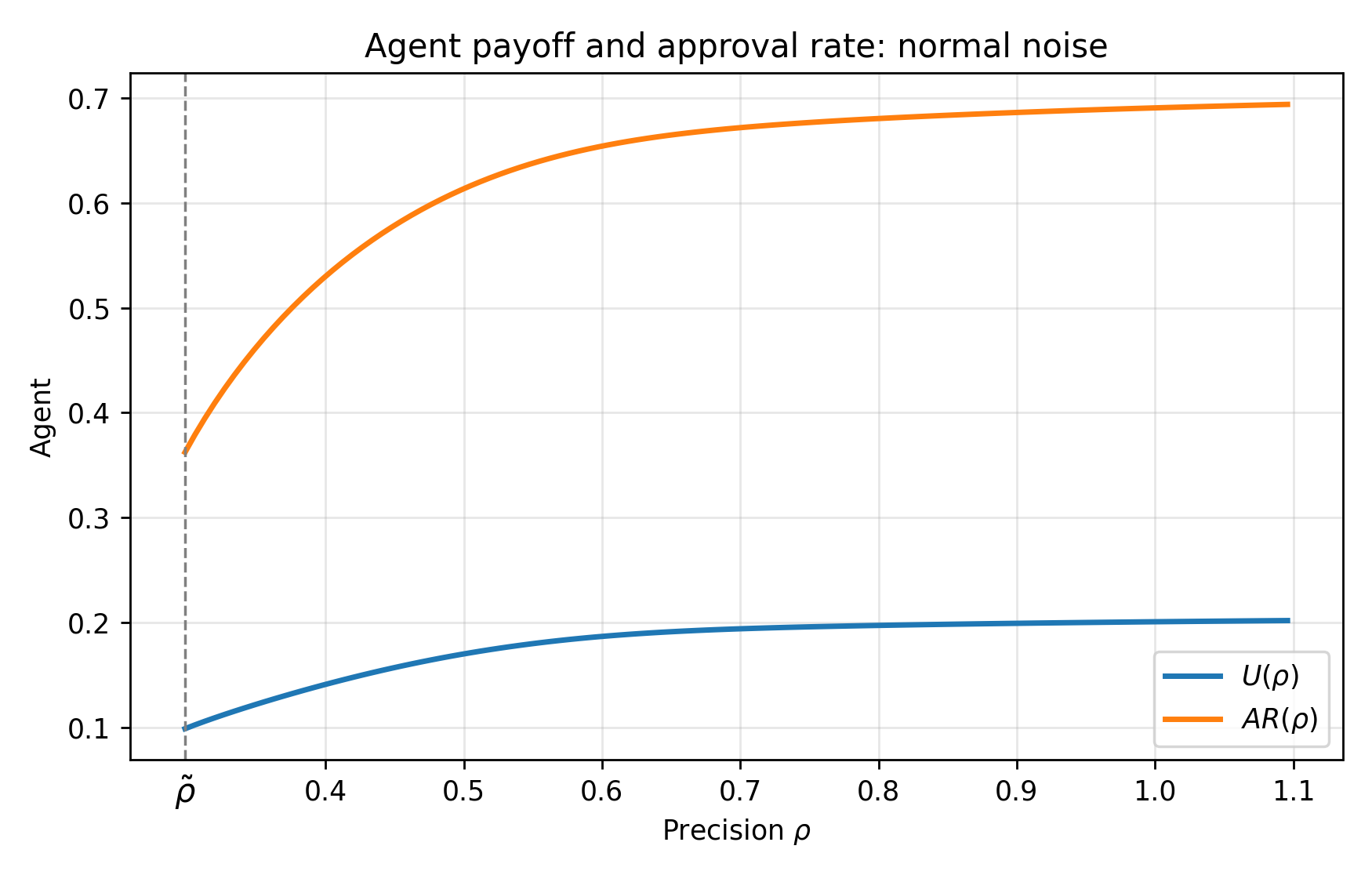}
        \caption{Normal noise.}
    \end{subfigure}
    \hfill
    \begin{subfigure}{0.48\textwidth}
        \centering
        \includegraphics[width=\linewidth]{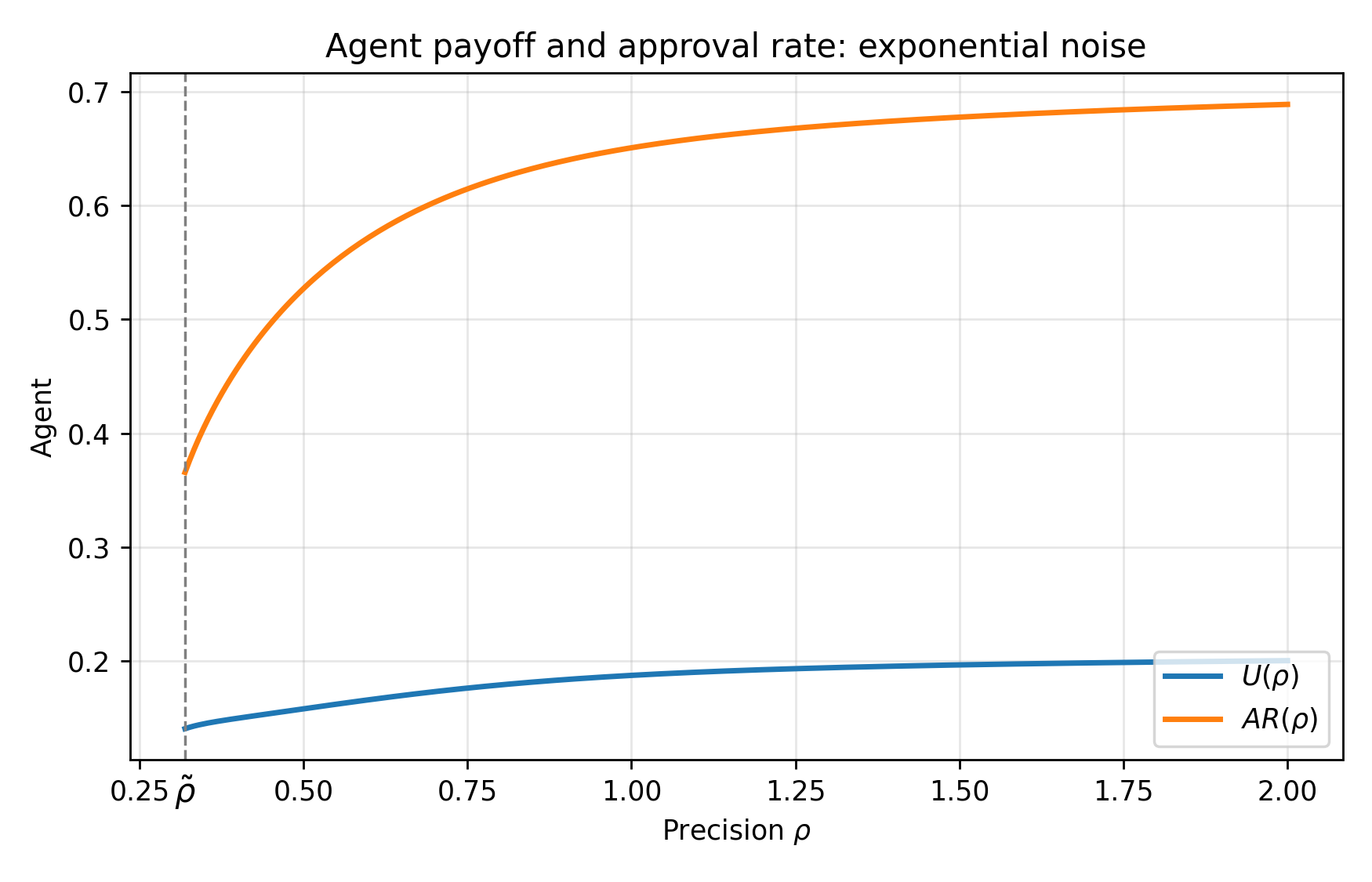}
        \caption{Exponential noise.}
    \end{subfigure}
    \caption{Agent payoff (blue) and approval rate (orange), where $\theta$ is taken to be uniformly distributed over $[5,15]$, with $v(\theta)=\theta-11$.}
    \label{fig:UAR-mono}
\end{figure}
Hence, in terms of both approval outcomes and payoffs (net of signaling costs), the agent on average benefits from higher precision, as illustrated in \Cref{fig:UAR-mono}. However, the effects are heterogeneous across types: low types face lower approval rates; high types face higher approval rates; and due to the strategic effect, some marginal types begin investing effort and get approved with higher probability. Higher precision may therefore exacerbate inequality across types. When precision is high, the gains to high types---via higher approval rates despite potentially higher standards and signaling costs---dominate the losses to low types, so average agent payoff increases in $\precision$.

Similar to \Cref{prop:non-monotone}, \Cref{prop:approval} relies on a high-precision asymptotic analysis ($\precision\to\infty$) that we develop to isolate the first-order effects of precision. 
The tail regularity condition $\lim_{z\to\infty}\frac{f(z)f''(z)}{[f'(z)]^2}=1$ ensures stable asymptotic behavior in the tail; in particular, it implies $\lim_{z\to\infty}[\frac{1-F(z)}{f(z)}]'=0$, meaning $\frac{1-F(z)}{f(z)}$ varies slowly in the tail.\footnote{The regularity condition is satisfied by many common distributions, e.g., exponential, generalized normal, and Weibull.}


\paragraph{Statistical Discrimination.}
The above results also have implications for statistical discrimination \`a~la \citet{phelps1972statistical}. As in the standard framework, agents from different groups (e.g., by gender or ethnicity) share the same distribution of productivity but differ in their access to signaling technologies. \Cref{prop:approval} implies that groups with noisier signaling technologies face lower approval rates and lower average payoffs.\footnote{This comparison remains valid even when the principal faces a binding quota across groups. In that case, it suffices to replace the principal's payoff from approval $v(\theta)$ with $v(\theta)-\lambda$, where $\lambda$ is the shadow cost of the quota constraint. The comparative statics of $AR$ and $U$ thus continue to hold as long as $\mathbb{E}[v(\theta)-\lambda]\leq 0$.} This constitutes group-level statistical discrimination against groups with less informative signaling technologies, consistent with classical statistical discrimination theory \citep{aigner1977statistical,chambers2021characterisation}.\footnote{Group-level statistical discrimination \citep[cf.][]{aigner1977statistical} arises when two groups with the same underlying productivity distribution but different signaling technologies receive different outcomes on average.}

However, \Cref{prop:non-monotone} reveals a reversal: the principal may actually favor groups with noisier technologies ex ante, precisely because their agents behave less strategically. If the principal solicits applications ex ante, she will prefer to target groups with noisier technologies. This provides a rationale for why institutions may wish to direct more favorable quotas or recruitment attention toward groups that are otherwise disfavored at the approval stage.

\begin{corollary}
    Suppose two groups share the same productivity distribution but differ in signaling precision, with $\precision_A>\precision_B>\bar\precision$. Then group~$B$ has a lower 
    approval rate and lower average payoff during screening. If, prior to screening, the principal has a single application slot and chooses which group to solicit, she strictly prefers group~$B$.
\end{corollary}

This reversal carries a broader methodological lesson: a disparity at a single stage need not constitute evidence of discrimination nor reflect a more global pattern, as outcomes depend on the full trajectory of interactions between groups and institutions. See \citet{bohren2025systemic} for a related discussion on the measurement of discrimination.

\subsection{Screening Efficiency and Information Transmission}
\label{sect:accuracy}
\paragraph{Screening Efficiency.}
The trade-off behind the non-monotonicity of $V(\precision)$ can be captured by the impact of the precision on screening efficiency, which is decomposed into Type I and Type II errors:
\begin{align*}
\alpha(\precision):=&\text{Pr}\Big(\text{rejecting types above $\tilde\theta$}\Big)=\int_{\tilde\theta}^{\overline{\theta}}g(\theta)F(-f^{-1}(1/(\precision\theta)))\mathrm{d}\theta,\\
\beta(\precision):=&\text{Pr}\Big(\text{approving types below $\tilde\theta$}\Big)\\
=&G(\hat\theta_-(\precision))[1-F(\hat\tau_-(\precision))]+\int_{\hat\theta_-(\precision)}^{\tilde\theta}g(\theta)[1-F(-f^{-1}(1/(\precision\theta)))]\mathrm{d}\theta.
\end{align*}
\begin{proposition}
\label{prop:error}
    The Type I error $\alpha(\precision)$ is decreasing in $\precision$, while the Type II error $\beta(\precision)$ is increasing in $\precision$ when $\precision$ is large. 
\end{proposition}
Precision always mitigates the Type I error. However, when precision is already very high, the strategic effect of precision is of the first order: a higher precision incentivizes the low-type agent to exert effort, leading to a greater Type II error and making the principal worse off. 

\paragraph{Information Transmission and Accuracy.}
As seen before, 
the strategic effect plays an important role behind the downside of precision. Suppose that instead the agent is non-strategic and commits to an increasing strategy $e(\theta)$ contingent on his types. Then according to \citet{lehmann1988comparing}, the higher the precision, the more accurate the signal, the higher the principal's expected payoff. 
This stands in contrast to our result and suggests that the agent's strategic response to precision must distort the accuracy of the information faced by the principal.

To investigate transmitted information, we adopt the quantile-function approach introduced by \citet{lehmann1988comparing}. Let $e^*(\theta;\precision)$ denote the agent's effort strategy in equilibrium $(\hat\tau_-,\hat\theta_-)$. Define $H_{\precision}(s|\theta):=F_\precision(s-e^*(\theta;\precision))=F(\precision(s-e^*(\theta;\precision)))$ as the equilibrium signal distribution conditional on the agent being type $\theta$. 

The defined $H_{\precision}(s|\theta)$ is the information structure faced by the principal that is endogenous in equilibrium. In general, these endogenous information structures with different $\precision$'s are not comparable by either the Blackwell order \citep{blackwell1951comparison,blackwell1953equivalent} or the accuracy order. Instead, we focus on information structures restricted to pairs of types, i.e., $\left(H_{\precision}(s|\theta),H_{\precision}(s|\theta')\right)_{s\in\mathbb{R}}$ for any $\theta,\theta'$. 

Fix $\theta$ and $\theta'$ such that $\theta<\theta'$. For information structures with binary types $\theta,\theta'$, $H$ is more accurate \citep{lehmann1988comparing} than $\tilde{H}$ if\footnote{With binary state, the accuracy order is equivalent to the Blackwell order, but much easier to work with when information structures satisfy MLRP.} 
\begin{align*}
H^{-1}(\tilde{H}(s|\theta)|\theta)&\leq H^{-1}(\tilde{H}(s|\theta')|\theta'),\quad\text{for all }s\in\mathbb{R},\\
\text{or equivalently,}\quad H(H^{-1}(p|\theta)|\theta')&\leq \tilde{H}(\tilde{H}^{-1}(p|\theta)|\theta'),\quad\text{for all }p\in[0,1].
\end{align*}

\begin{proposition}
\label{prop:accuracy}
    For $\precision>\precision'>\tilde{\precision}$, 
    \begin{itemize}
        \item[(a)] For $\theta<\hat\theta_-(\precision)<\tilde\theta$ and $\theta'>\tilde\theta$, the transmitted information about $\theta,\theta'$ in equilibrium is more accurate under $\precision$ than under $\precision'$; 
        \item[(b)] For $\theta\in(\hat\theta_-(\precision),\tilde\theta)$ and $\theta'>\tilde\theta$, the transmitted information about $\theta,\theta'$ in equilibrium is less accurate under $\precision$ than under $\precision'$.
    \end{itemize}
\end{proposition}
Part (a) is consistent with the plain intuition that a more precise screening process should generate more accurate information about the agent; this is true when the principal focuses on screening sufficiently bad types ($\theta<\hat\theta_-(\precision)$) from good types ($\theta'>\tilde\theta$). In contrast, according to part (b), when precision is increased, it gets more difficult to distinguish moderately bad types ($\theta\in(\hat\theta_-(\precision),\tilde\theta)$) and good types ($\theta'>\tilde\theta$) because a higher precision incentivizes moderate types to invest more effort.


\section{Discussion}
\label{sect:discussion}
This section provides further discussions on the mechanisms behind our results. \Cref{sect:absorbing,sect:principal-commit} reveal that the pitfall of precision is not only due to the agent's strategic signaling behavior, but also driven by the principal's lack of commitment.

In particular, \Cref{sect:absorbing} shows that if the principal can choose the precision of the screening technology but has no commitment, she will exactly fall into the pitfall of precision. On the other hand, \Cref{sect:principal-commit} shows commitment power can help the principal overcome the pitfall.

Finally, 
\Cref{sect:simulation} shows, with numerical simulations, that the pitfall extends beyond linear costs. What matters is the existence of an extensive margin in signaling, which can arise under general convex costs or binary effort.


\subsection{The ``Absorbing'' Pitfall}
\label{sect:absorbing}
Suppose that the principal can choose among different screening technologies with different precision levels in $[\tilde{\precision},\precision_0]$, where $\precision_0\gg\bar{\precision}$ and $\bar{\precision}$ is the precision level identified in \Cref{prop:non-monotone}(b).

If the principal can publicly commit to a precision level before the agent exerts effort and sends signals, then she will choose the ex ante optimal precision
\[
\hat{\precision}:=\arg\max_{\precision\in[\tilde{\precision},\precision_0]}V(\precision).
\]
By \Cref{prop:non-monotone}, we know  $\hat{\precision}\in[\tilde{\precision},\bar{\precision}]$.

However, if the principal has no such commitment power and she privately chooses the precision level, then the agent must form an expectation about what precision the principal will choose. It is easy to establish that the principal will choose the technology with the highest precision $\precision_0$ and anticipating this, players will play some equilibrium in the baseline model associated with $\precision_0$. As a result, by \Cref{prop:eqm} and \Cref{cor:pareto}, the principal can get at most $V(\precision_0)<V(\bar{\precision})\leq V(\hat{\precision})$ and thus fall into the pitfall identified in \Cref{prop:non-monotone}.

\begin{proposition}
\label{prop:absorbing}
    If the principal privately chooses a precision level in $[\tilde{\precision},\precision_0]$, then it is dominant to choose $\precision=\precision_0$ and the equilibrium payoff is at most $V(\precision_0)<V(\hat{\precision})$.
\end{proposition}

Therefore, if the principal can costlessly choose among different screening technologies with different precision levels, but she cannot publicly commit to the precision choice, then falling into the pitfall of precision is unavoidable. 

Indeed, the principal's lack of commitment is an important source of the pitfall of precision. Having some commitment power over the precision choice (e.g. having costs for precision, as a way to gain commitment) can surely help get the principal out of this dilemma. Interestingly, as will be discussed later in \Cref{sect:principal-commit}, if the principal can publicly commit to an approval standard before the agent exerts effort, the pitfall of precision can also be overcome.

\subsection{Strategic Principal and Low Standards}
\label{sect:principal-commit}
Recall that the strategic effect of precision is negative because the marginally incentivized types are bad types as $\hat\theta_-(\precision)<\tilde\theta$. The agent optimally adopts this threshold $\hat\theta_-(\precision)$ for they in equilibrium anticipate that the principal will use the approval standard $\hat\tau_-(\precision)$. 

Although $\hat\tau_-(\precision)$ is a best response to $\hat\theta_-(\precision)$, this equilibrium standard is too low from the ex ante perspective. Should the principal be able to publicly commit to an approval standard, she could set a high standard to deter all bad types from investing effort. As such, a higher precision can only incentivize more good types to exert effort, so the pitfall of precision no longer exists. 

\paragraph{Committing to An Approval Standard.}
When the principal can publicly commit to an approval standard $\hat\tau$, her problem is as follows:\footnote{Since our focus is on the effect of precision, for simplicity, we restrict the principal to choosing an approval standard. In principle, she could adopt more flexible, potentially randomized, approval rules. Note that the principal would still benefit from higher precision: she can always inject additional noise to replicate any optimal approval rule with less precise signals.}
\[
\overline{V}(\precision):=\max_{\hat\tau}  \int_{\underline\theta}^{\widehat\Theta(\hat\tau;\precision)} v(\theta)g(\theta)\mathrm{d}\theta\cdot[1-F(\hat\tau)]+\int_{\widehat\Theta(\hat\tau;\precision)}^{\overline\theta} v(\theta)g(\theta)[1-F(-f^{-1}(1/(\precision\theta)))]\mathrm{d}\theta.
\] 
Let $\hat\tau^*(\precision)$ be the optimal solution to this problem, hence the threshold type is $\hat\theta^*(\precision):=\widehat\Theta(\hat\tau^*(\precision);\precision)$. We have the following results.
\begin{proposition}
\label{prop:commit}
    With commitment, the principal optimally commits to a higher approval standard, $\hat\tau^*(\precision)>\hat\tau_-(\precision)$, incentivizes only good types to exert effort, $\hat\theta^*(\precision)>\tilde\theta$, and benefits from commitment, $\overline{V}(\precision)\geq V(\precision)$. Moreover, $\overline{V}(\precision)$ is increasing in $\precision$.
\end{proposition}
When the principal makes approval decisions after observing signals, she does not take into account how the approval decision affects agents' incentives. In equilibrium, the principal adopts a low standard to approve as many good types as possible, which however incentivizes some moderately bad types to invest effort. A higher precision encourages even more bad types.

In contrast, if the principal has commitment power, then she will optimally set a high standard and deter all bad types from investing effort, though this means that she has to also reject many good types. With commitment, the principal always benefits from higher precision since the strategic effect of precision is no longer harmful. Hence, commitment can address the pitfall of precision.

\paragraph{Which Kind of Commitment?} We have so far discussed two forms of commitment: over the precision level and over the approval standard. Since $\overline{V}(\precision)$ is increasing in $\precision$, the principal always prefers the highest available precision level when she can commit to an approval standard, and thus gains nothing from additionally being able to commit to the precision level. Since $\overline{V}(\precision)\geq V(\precision)$, commitment over the approval standard is unambiguously more valuable to the principal than commitment over the precision level.

\paragraph{Competition and Commitment.}
    Full commitment can be hard to guarantee in practice. However, in many contexts, the principal is endowed with some partial commitment power by the nature of the problem. For example, competition among agents provides such a situation.
    
    For simplicity, let us assume there are two agents competing for one position and the principal has to choose one of them to fill in the position. Consider symmetric equilibrium where agents use the same strategy; by the single-crossing property, the symmetric strategy must be increasing. Then no matter deciding before or after seeing the signals, it is optimal for the principal to choose one agent if and only if his signal is better than the other. In this sense, competition between agents makes the principal as if committed. 
    
    The principal's selection rule generates a noisy ``winner-take-all'' contest \citep[see][]{tullock1980efficient,lazear1981rank} between the two agents. The symmetric equilibrium strategy $\bar{e}(\theta;\precision)$ is determined by
\[
\bar{e}(\theta;\precision)\in\arg\max_{e\geq0}\text{ }\int_{\underline\theta}^{\overline\theta}g(\theta')\tilde{F}\left(\precision(e-\bar{e}(\theta';\precision))\right)\mathrm{d}\theta'-\frac{e}{\theta},
\]
where $\tilde{F}$ is the distribution of $\epsilon_1-\epsilon_2$ and $\epsilon_1,\epsilon_2$ are independent noises distributed according to $F$.

In the limit of $\precision\to\infty$, the contest becomes an all-pay auction that leads to the efficient allocation. In detail, as $\precision\to\infty$, $\bar{e}$ converges to $b^*(\theta)=\int_{\underline\theta}^\theta g(\theta')\theta'\mathrm{d}\theta'$ and the principal's payoff converges to $\mathbb{E}[v(\theta_1)|\theta_1\geq\theta_2]$, which is the best she can get provided that she has to select one agent. As a result, there is no pitfall of precision.

If we adopt the population interpretation of the model, then it is also possible to consider competition within the population for a fixed quota \citep[cf.][]{adda2024grantmaking} where the principal must approve a given amount of agents. The fixed quota can also give the principal some commitment power and thus may as well alleviate the pitfall of precision.

\subsection{Beyond Linear Costs}
\label{sect:simulation}

The baseline model assumes linear effort costs, which yield a sharp zero-versus-positive effort cutoff. This section shows that the pitfall of precision extends beyond linear costs. What matters is the existence of an extensive margin: as precision increases, additional types switch discretely from low to high effort. We show how this margin can still arise under general convex costs and under binary effort.

\paragraph{General Convex Costs.}

Suppose the agent's effort cost is $C(e,\theta)$, where $C(0,\theta)=0$, $C_e\geq0$, $C_{ee}\geq0$, and $C_{e\theta}\leq0$ (higher types have lower marginal costs). Given approval standard $\tau$, type $\theta$ chooses $e\geq0$ to maximize
\[
    1-F(\precision(\tau-e))-C(e,\theta).
\]
The first-order condition for an interior optimum is
\[
    \precision f(\precision(\tau-e))=C_e(e,\theta).
\]
The left-hand side is single-peaked in $e$ around $\tau$, while the right-hand side is increasing in $e$. With general convex costs, the two may intersect multiple times, producing multiple local optima, possibly including the corner $e=0$. This significantly complicates the analysis, which is why we focus on linear costs in the baseline.

For ease of understanding, suppose the agent's problem has two relevant local optima: a low-effort choice $e_L(\theta;\tau,\precision)$, typically below the standard and possibly zero, and a high-effort choice $e_H(\theta;\tau,\precision)$. The agent chooses the high-effort branch if and only if
\[
    1-F(\precision(\tau-e_H))-C(e_H,\theta)
    \geq
    1-F(\precision(\tau-e_L))-C(e_L,\theta).
\]
By the single-crossing property implied by $C_{e\theta}\leq0$, the agent's best response is characterized by a threshold $\hat\theta(\tau,\precision)$: types above $\hat\theta$ choose the high-effort branch, and types below choose the low-effort branch.

The principal's indifference condition at the approval standard is then
\[
    \int_{\underline\theta}^{\hat\theta}
    v(\theta)g(\theta)f_\precision(\tau-e_L(\theta;\tau,\precision))\mathrm{d}\theta
    +\int_{\hat\theta}^{\overline\theta}
    v(\theta)g(\theta)f_\precision(\tau-e_H(\theta;\tau,\precision))\mathrm{d}\theta=0.
\]
Together with the agent's threshold condition, this characterizes the equilibrium under general convex costs. The structure is analogous to the baseline model, though the equilibrium is significantly harder to analyze.

We numerically simulate the equilibrium under quadratic costs $C(e,\theta)=e^2/(2\theta)$. As \Cref{fig:quadratic-effort-profile} shows, agents indeed switch between two local optima at high precision. Importantly, \Cref{fig:non-linear} shows that the principal's payoff remains non-monotone in precision: the pitfall survives under convex costs. The mechanism is the same as in the baseline model: higher precision raises the value of switching to the high-effort branch, inducing more unqualified types to increase effort on the extensive margin.

\begin{figure}[!htbp]
    \centering
    \includegraphics[width=.5\textwidth]{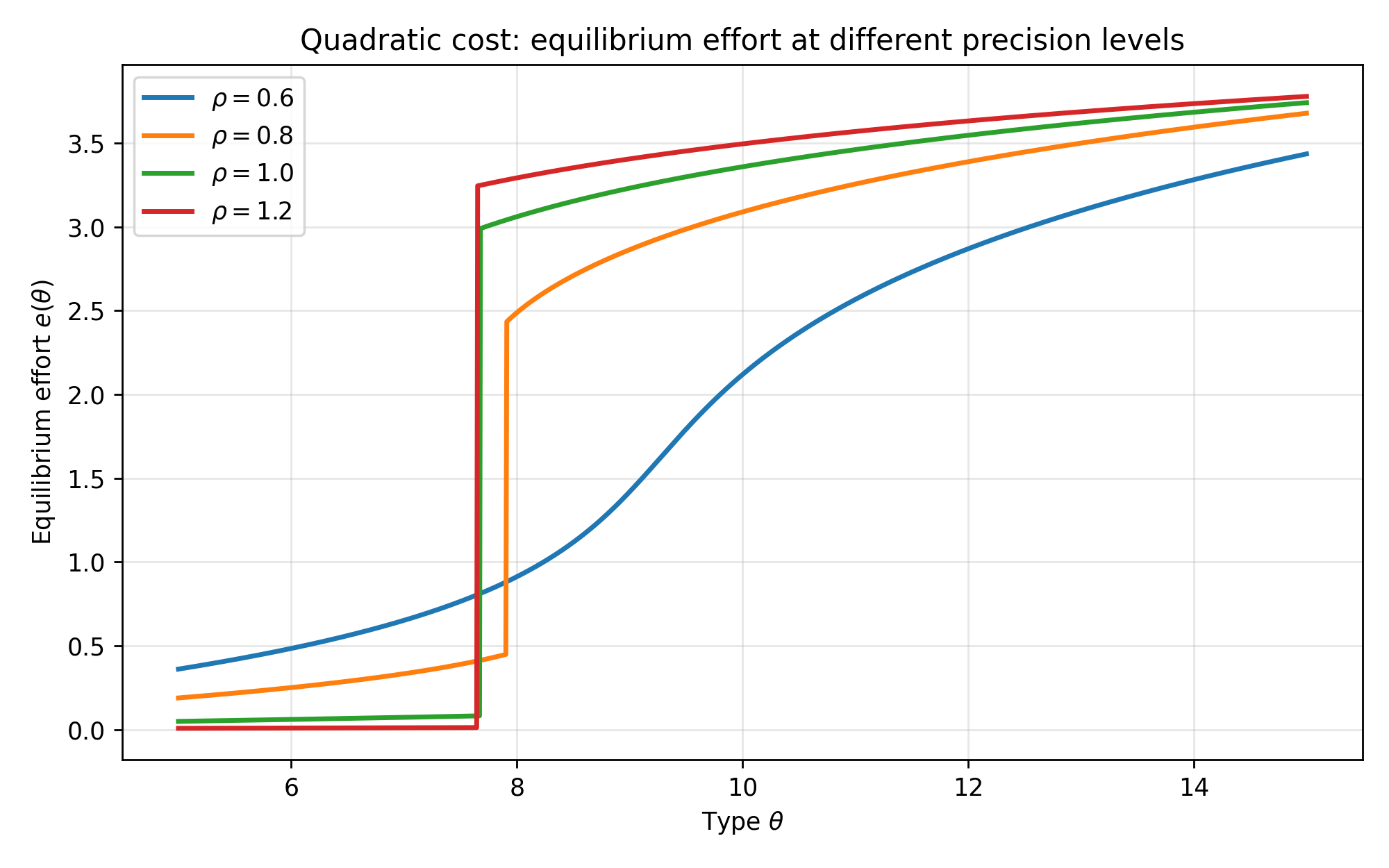}
    \caption{Equilibrium effort under quadratic costs: $C(e,\theta)=e^2/(2\theta)$, 
    normal noise, $\theta\sim U[5,15]$, and $v(\theta)=\theta-11$.}
    \label{fig:quadratic-effort-profile}
\end{figure}

\begin{figure}[!htbp]
    \centering
    \begin{subfigure}{.48\textwidth}
        \centering
        \includegraphics[width=\textwidth]{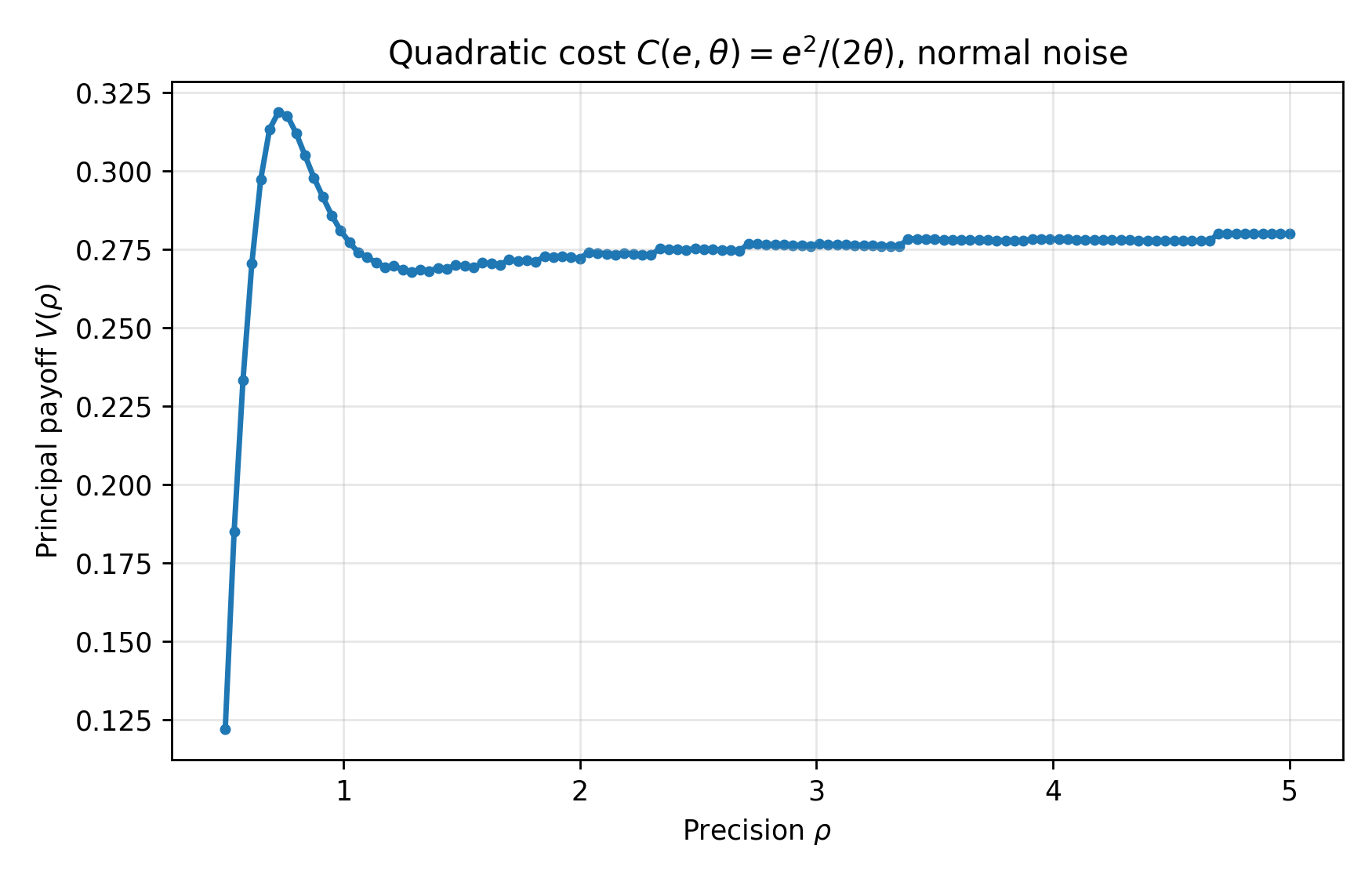}
        \caption{Quadratic effort cost: $C(e,\theta)=e^2/(2\theta)$.}
    \end{subfigure}
    \begin{subfigure}{.48\textwidth}
        \centering
        \includegraphics[width=\textwidth]{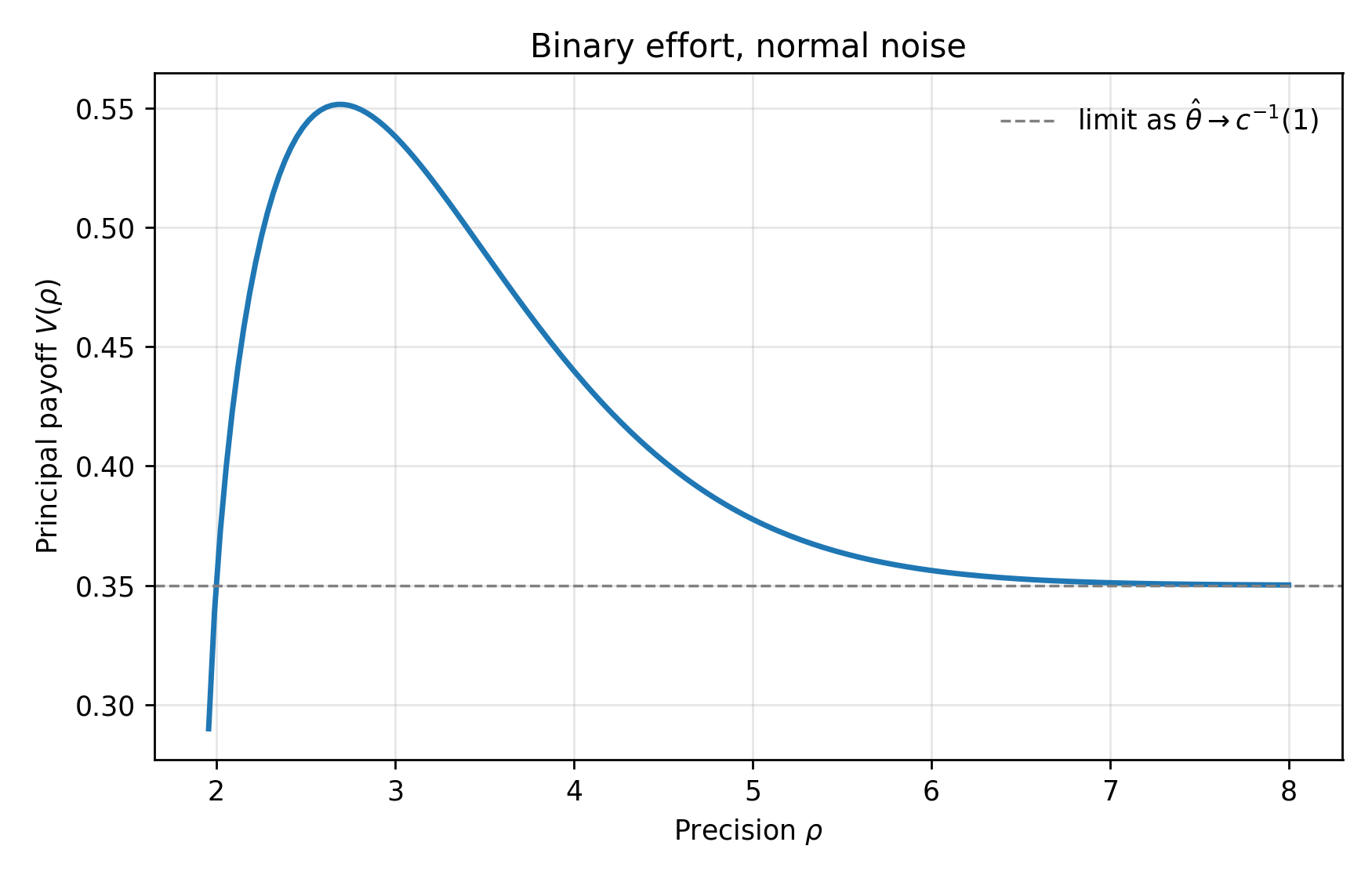}
        \caption{Binary effort: $e\in\{0,1\}$ and $c(\theta)=8/\theta$.}
    \end{subfigure}
    \caption{Principal payoff under normal noise, $\theta\sim U[5,15]$, 
    and $v(\theta)=\theta-11$.}
    \label{fig:non-linear}
\end{figure}

\paragraph{Binary Effort.}

A particularly transparent, though mechanical, way to isolate the extensive margin is to make effort binary. Suppose that $e\in\{0,\bar e\}$, $C(0,\theta)=0$, and $C(\bar e,\theta)=c(\theta)$, where $c(\theta)$ is strictly decreasing. Given standard $\tau$, the benefit of choosing high effort is
\[
    \Delta(\tau,\precision)
    =\Pr(\bar e+\epsilon/\precision\geq\tau)
    -\Pr(\epsilon/\precision\geq\tau)
    =F(\precision\tau)-F(\precision(\tau-\bar e)).
\]
The agent chooses high effort if and only if $c(\theta)\leq\Delta(\tau,\precision)$, so the agent's strategy is characterized by a cutoff
\[
    \hat\theta(\tau,\precision):=c^{-1}(\Delta(\tau,\precision)).
\]
Types above $\hat\theta$ choose $\bar e$; types below choose zero effort.

The principal's indifference condition is similar to before and, together with the cutoff condition, characterizes the equilibrium.

The effect of precision is especially transparent. If the standard lies between $0$ and $\bar e$ (as it does in equilibrium), then $\Delta(\tau,\precision)$ is increasing in $\precision$ and converges to $1$ as $\precision\to\infty$, so $\hat\theta$ decreases toward $c^{-1}(1)$. If $c^{-1}(1)<\tilde\theta$, the marginal types induced to choose high effort at high precision are bad types, harming the principal. \Cref{fig:non-linear} confirms this numerically.

These examples clarify the role of the extensive margin. The pitfall of precision does not hinge on the linear cost assumption; linear costs merely make the extensive margin analytically sharp. More generally, whenever higher precision induces marginal low types to switch from low to high effort, the principal can be harmed by greater precision.

\section{Conclusion}
This paper has shown that greater screening precision can backfire in  noisy signaling environments with costly, strategic effort. Higher precision reduces random errors in selection. But at the same time, it intensifies incentives for marginally unqualified agents to engage in wasteful signaling, hence reducing screening efficiency and lowering the principal's welfare once precision is sufficiently high. We refer to this paradox as the pitfall of precision. This mechanism we identify also highlights the institutional importance of commitment power: if principals could commit ex ante to an optimal approval standard, greater precision would unambiguously improve welfare.


The pitfall of precision sheds light on screening technologies and their heterogeneous impact across groups. Our mechanism reveals a fundamental welfare trade-off beyond screening costs: even if more precise technology were free, it need not be desirable, because it could reduce welfare due to strategic distortions, and increase inequalities when groups differ in ability to exploit technology, effectively subjecting them to different precisions. Our results offer a rationale for why screening institutions, ranging from tax enforcement and welfare eligibility to college admissions, may in practice deliberately limit precision or favor disadvantaged groups with noisier technologies. 
Our analysis also points to a broader research agenda: understanding how  noise and strategic signaling interact to shape welfare, inequality, and institutional design.

\newpage

\appendix
\section{Omitted Results and Proofs}
\label{app}

\paragraph{Notation.}
Throughout the appendix, we work with the noise level $\sigma:=1/\precision$ rather than precision $\precision$, as the analysis is most convenient and transparent in $\sigma$, especially for comparative statics and asymptotic analysis at high precision ($\sigma=1/\precision\to0$). Let $f_\sigma(z):=\frac{1}{\sigma}f(\frac{z}{\sigma})$ denote the pdf of $\sigma\epsilon$, parametrized by $\sigma$.

Equilibrium objects in the main text are written as functions of precision $\precision$; in the appendix we write the corresponding objects as functions of $\sigma$. For each equilibrium object $X$, define
\[
X^\sigma(\sigma):=X^\precision(1/\sigma),
\]
so that $X^\precision(\precision)=X^\sigma(1/\precision)$ and
\[
\frac{\mathrm{d} X^\precision}{\mathrm{d}\precision}(\precision)=-\frac{1}{\precision^2}\frac{\mathrm{d} X^\sigma}{\mathrm{d}\sigma}(1/\precision).
\]
Comparative statics in $\sigma$ translate directly to those in $\precision$ via this identity, with opposite signs. Superscripts $\sigma$ are dropped when there is no confusion.

\subsection{Formal Analysis and Proofs from \Cref{sect:eqm}}
\label{app:eqm-proof}

Since the agent's payoff satisfies the single-crossing property (SCP) in $(e,\theta)$, his equilibrium strategy $e(\theta)$ must be increasing. The monotonicity of $e(\theta)$ and the MLRP of $h_\sigma(s|e)$ together imply that the principal's posterior belief is increasing in $s$ in the first-order stochastic dominance order. Since $v$ is increasing in types, her equilibrium strategy must be a cutoff rule: approve if and only if $s>\tau$ for some $\tau\in\overline{\mathbb{R}}$.

\paragraph{Agent's Best Response.}
Given any cutoff $\tau$, type $\theta$'s payoff from effort $e$ is
\[
\text{Pr}(e+\sigma\cdot\epsilon>\tau)-e/\theta=1-F_\sigma(\tau-e)-e/\theta.
\]
It is easy to see that this payoff is convex for $e\in[0,\tau]$ and concave for $e\geq\tau$. Therefore, the agent's best response is either the corner solution $e=0$, or the interior solution $e=\hat{e}(\theta;\tau):=[\tau+f_{\sigma}^{-1}(1/\theta)]^+=[\tau+\sigma f^{-1}(\sigma/\theta)]^+$, determined by the first order condition that $1/\theta=f_\sigma(\tau-e)=f_\sigma(e-\tau)$ with $e\geq\tau$
.\footnote{Here $(x)^+:=\max\{x,0\}$ for $x\in\mathbb{R}$. Note that the interior candidate effort level, $e=[\tau+f_\sigma^{-1}(1/\theta)]^+$, exists only when $\sigma/\theta<f(0)$. For the consistency of notations, we adopt the convention that $f^{-1}(p):=0$ for $p>f(0)$.} 

Due to the SCP of the agent's payoff, a threshold type $\hat\theta(\tau;\sigma)\in[\underline\theta,\overline{\theta}]$ exists such that all higher types strictly prefer $e=[\tau+f_{\sigma}^{-1}(1/\theta)]^+$ and all lower types strictly prefer $e=0$. The threshold $\hat\theta$ is determined by the indifference condition (\ref{eq:S-BR}) when $\hat\theta$ is interior, which specializes to
\[
1-F_\sigma(\tau-0)=1-F_\sigma(\tau-[\tau+f_\sigma^{-1}(1/\hat\theta)]^+)-[\tau+f_\sigma^{-1}(1/\hat\theta)]^+/\hat\theta.
\]
By the symmetry of $f$, it can be rewritten as follows:
\begin{equation}
\label{eq:S-BR-app}
    F_\sigma(0-\tau)=F_\sigma(\max\{f_\sigma^{-1}(1/\hat\theta),-\tau\})-[\tau+f_\sigma^{-1}(1/\hat\theta)]^+/\hat\theta. \tag{A-BR'}
\end{equation}
Instead, for corner solutions: for $\hat\theta=\overline\theta$, we must have the left-hand side (LHS) is greater than the right-hand side (RHS) in \Cref{eq:S-BR-app}; symmetrically for $\hat\theta=\underline\theta$, we should have the LHS smaller than the RHS.

Let $\hat\tau:=\tau/\sigma$. Define $\widehat\Theta:\overline{\mathbb{R}}\to[\underline\theta,\overline\theta]$ as
\begin{equation}
\label{eq:S-BR-function}
    \widehat\Theta(\hat\tau;\sigma):=\inf\left\{\hat\theta\in[\underline\theta,\overline\theta]:F(-\hat\tau)< F(\max\{f^{-1}(\sigma/\hat\theta),-\hat\tau\})-[\hat\tau+f^{-1}(\sigma/\hat\theta)]^+\sigma/\hat\theta\right\}\tag{A-BR''}
\end{equation}
with the convention that $\widehat\Theta(\hat\tau;\sigma):=\overline\theta$ if the corresponding set is empty. Here for expositional convenience later, we abuse notation and write \(\widehat\Theta\) as a function of \(\hat\tau=\tau/\sigma\) rather than $\tau$ as in the main context; obviously, $\widehat\Theta(\tau)=\widehat\Theta(\sigma\hat\tau)$.

\Cref{lemma:Theta-increasing} in the main text follows from the following result.

\begin{lemma}
\label{lemma:S-BR-increasing}
    For a given $\hat\tau$, if $\widehat\Theta(\hat\tau)\in(\underline\theta,\overline\theta)$, then $\hat\theta=\widehat\Theta(\hat\tau)$ is the unique type that is indifferent between $e=0$ and $e=\sigma[\hat\tau+f^{-1}(\sigma/\theta)]$; if $\widehat\Theta(\hat\tau)=\underline\theta$ (resp., $\overline\theta$), then all types prefer positive effort (resp., all types prefer zero effort). 
    
    For $\hat\tau\geq0$ such that $\widehat\Theta(\hat\tau)\in(\underline\theta,\overline{\theta})$, $\widehat\Theta(\hat\tau)$ is strictly increasing in $\hat\tau$; for $\hat\tau<0$ such that $\widehat\Theta(\hat\tau)\in(\underline\theta,\overline{\theta})$, $\widehat\Theta(\hat\tau)=\sigma/f(-\hat\tau)$, which is strictly decreasing in $\hat\tau$.
\end{lemma}
\begin{proof}
First, when $\hat\tau>0$, the indifference condition (\ref{eq:S-BR-app}) can be rewritten as
\[
0=F(f^{-1}(\sigma/\hat\theta))-[\hat\tau+f^{-1}(\sigma/\hat\theta)]\sigma/\hat\theta-F(-\hat\tau)=:Q(\hat\theta,\hat\tau).
\]
Note that $Q$ is strictly increasing in $\hat\theta$ since
\[
\frac{\partial Q}{\partial\hat\theta}=[\hat\tau+f^{-1}(\sigma/\hat\theta)]\sigma/\hat\theta^2>0.
\]
Therefore, for any $\hat\tau>0$, as long as $Q(\underline\theta,\hat\tau)< 0< Q(\overline\theta,\hat\tau)$, which is equivalent to $\widehat\Theta(\hat\tau)\in(\underline\theta,\overline\theta)$ by definition, there exists a unique $\hat\theta\in(\underline\theta,\overline\theta)$ that solves \Cref{eq:S-BR-app}, which is exactly given by $\widehat\Theta(\hat\tau)$. If instead $Q(\underline\theta,\hat\tau)\geq0$ (resp., $Q(\overline\theta,\hat\tau)\leq0$), equivalently $\widehat\Theta(\hat\tau)=\underline\theta$ (resp., $\widehat\Theta(\hat\tau)=\overline\theta$), then all types prefer positive effort (resp., all types prefer zero effort).

Then for $\hat\tau\leq0$, the agent of type $\theta$ wants to exert positive effort if and only if $\hat\tau+f^{-1}(\sigma/\theta)\geq0$. Hence $\widehat\Theta(\hat\tau)=\inf\{\theta\in[\underline\theta,\overline\theta]:\hat\tau+f^{-1}(\sigma/\theta)\geq0\}=\min\{\max\{\sigma/f(-\hat\tau),\underline\theta\},\overline\theta\}$, which is exactly the cutoff type above which all types prefer $e=\sigma[\hat\tau+f^{-1}(\sigma/\theta)]$ and below which all types prefer $e=0$.

Finally, the monotonicity of $\widehat\Theta$ in $\hat\tau\geq0$ is by the Implicit Function Theorem and the fact that when $Q(\hat\theta,\hat\tau)=0$, the partial derivative of $Q$ w.r.t. $\hat\tau$ is negative when $\hat\tau\geq0$:
\[
\frac{\partial Q}{\partial\hat\tau}\Big|_{Q(\hat\theta,\hat\tau)=0}=-[\sigma/\hat\theta-f(0-\hat\tau)]\Big|_{Q(\hat\theta,\hat\tau)=0}<0,
\]
where the inequality is because otherwise $Q(\hat\theta,\hat\tau)=\int_{-\hat\tau}^{f^{-1}(\sigma/\hat\theta)}[f(z)-\sigma/\hat\theta]\mathrm{d}z>0$, contradicting to $Q(\hat\theta,\hat\tau)=0$. Instead, when $\hat\tau<0$ and $\widehat\Theta(\hat\tau)\in(\underline\theta,\overline{\theta})$, $\widehat\Theta(\hat\tau)=\sigma/f(-\hat\tau)$ which is strictly decreasing in $\hat\tau$. 
\end{proof}

\paragraph{Principal's Best Response.}

Given the above cutoff structure of the agent's best response --- $e(\theta)=0$ for $\theta<\hat\theta$ and $e(\theta)=\tau+f_{\sigma}^{-1}(1/\theta)$ for $\theta\geq\hat\theta$ --- the principal's expected payoff from using $\tau$ is
\[
W(\tau;\hat\theta,\sigma):=\int_{\underline\theta}^{\hat\theta} v(\theta)g(\theta)\mathrm{d}\theta\cdot [1-F_\sigma(\tau-0)]+\int_{\hat\theta}^{\overline\theta} v(\theta)g(\theta)[1-F_\sigma(\tau-e(\theta))]\mathrm{d}\theta.
\]
By the monotonicity of $e(\theta)$ and the MLRP and log-concavity of $f_\sigma(s-e)$, $W$ is concave in $\tau$. Hence, the following first-order condition is both necessary and sufficient for any interior $\tau$ to be one of the principal's best responses:
\[
-f_\sigma(\tau-0)\int_{\underline\theta}^{\hat\theta} v(\theta)g(\theta)\mathrm{d}\theta-\int_{\hat\theta}^{\overline\theta} v(\theta)g(\theta)f_\sigma(\tau-e(\theta))\mathrm{d}\theta=0.
\]
In equilibrium $f_\sigma(\tau-e(\theta))=1/\theta$ for $\theta>\hat\theta$, hence the equilibrium condition for the principal can be simplified to:
\begin{equation}
-f_\sigma(\tau-0)\int_{\underline\theta}^{\hat\theta} v(\theta)g(\theta)\mathrm{d}\theta-\int_{\hat\theta}^{\overline\theta} v(\theta)g(\theta)/\theta\mathrm{d}\theta=0.\tag{P-BR}
\end{equation}

\Cref{eq:R-BR} can possibly hold only when $\hat\theta>\theta^\dagger$, with $\theta^\dagger$ defined in \Cref{eq:thetadagger}. If $\hat\theta\leq\theta^\dagger$, raising the approval standard only has benefits, 
so the principal's best response is $\tau=+\infty$. Note that $\theta^\dagger\in(\underline\theta,\tilde\theta)$ when $\mathbb{E}[v(\theta)]\leq0$.

For $\hat\theta\in[\underline\theta,\overline\theta]$, define
\begin{equation}
\label{eq:ratio}    
    R(\hat\theta):=\bigg(\int_{\hat\theta}^{\overline\theta} v(\theta)g(\theta)/\theta\mathrm{d}\theta\bigg)\bigg/\bigg(\int_{\underline\theta}^{\hat\theta} v(\theta)g(\theta)\mathrm{d}\theta\bigg).
\end{equation}
Therefore, \Cref{eq:R-BR} can be rearranged as\footnote{\label{fn:valid}Notice that $\tau=\pm T(\hat\theta)$ only characterizes the equilibrium condition when $-\sigma R(\hat\theta)\leq f(0)$. When $-\sigma R(\hat\theta)>f(0)$, \Cref{eq:R-BR} is not well-defined because the LHS is always negative. This is not an issue if we only consider relatively small noise, e.g., with $\sigma\leq-f(0)/R(\tilde\theta)$.}
\[
\tau=\pm\sigma f^{-1}(-\sigma R(\hat\theta))=: \pm T(\hat\theta;\sigma).
\]

\Cref{lemma:T-Ushaped} follows from the below properties of $R$ (in particular, \Cref{lemma:R-Ushaped}).
\begin{lemma}
\label{lemma:1/theta+R}
    When $\mathbb{E}[v(\theta)]\leq0$, $1/\hat\theta+R(\hat\theta)>0$ for all $\hat\theta\in[\underline\theta,\overline\theta]$.
\end{lemma}
\begin{proof}
    Observe that when $\hat\theta\in[\tilde\theta,\overline\theta]$,
    \begin{align*}
        \frac{1}{\hat\theta}\int_{\underline\theta}^{\hat\theta} v(\theta)g(\theta)\mathrm{d}\theta+\int_{\hat\theta}^{\overline\theta} v(\theta)g(\theta)/\theta\mathrm{d}\theta<&\frac{1}{\hat\theta}\int_{\underline\theta}^{\hat\theta} v(\theta)g(\theta)\mathrm{d}\theta+\frac{1}{\hat\theta}\int_{\hat\theta}^{\overline\theta} v(\theta)g(\theta)\mathrm{d}\theta\\
        =&\frac{1}{\hat\theta}\mathbb{E}[v(\theta)]\leq0.
    \end{align*}
    When $\hat\theta<\tilde\theta$, instead
    \begin{align*}
        \frac{1}{\hat\theta}\int_{\underline\theta}^{\hat\theta} v(\theta)g(\theta)\mathrm{d}\theta+&\int_{\hat\theta}^{\overline\theta} v(\theta)g(\theta)/\theta\mathrm{d}\theta<\frac{1}{\tilde\theta}\int_{\underline\theta}^{\hat\theta} v(\theta)g(\theta)\mathrm{d}\theta+\int_{\hat\theta}^{\tilde\theta} v(\theta)g(\theta)/\theta\mathrm{d}\theta\\
        &\qquad\qquad\qquad\qquad\qquad+\int_{\tilde\theta}^{\overline\theta} v(\theta)g(\theta)/\theta\mathrm{d}\theta\\
        \leq&\frac{1}{\tilde\theta}\int_{\underline\theta}^{\hat\theta} v(\theta)g(\theta)\mathrm{d}\theta+\frac{1}{\tilde\theta}\int_{\hat\theta}^{\tilde\theta} v(\theta)g(\theta)\mathrm{d}\theta+\frac{1}{\tilde\theta}\int_{\tilde\theta}^{\overline\theta} v(\theta)g(\theta)\mathrm{d}\theta\\
        =&\frac{1}{\tilde\theta}\mathbb{E}[v(\theta)]\leq0.
    \end{align*}
    As a result, in either case, we have $1/\hat\theta+R(\hat\theta)>0$.
\end{proof}

\begin{lemma}
\label{lemma:R-Ushaped}
    When $\mathbb{E}[v(\theta)]\leq0$, $R(\hat\theta)$ is strictly decreasing for $\hat\theta<\tilde\theta$ and increasing for $\hat\theta\geq\tilde\theta$.
\end{lemma}
\begin{proof}
Differentiating $R$, we can derive
    \[
R'(\hat\theta)
=-v(\hat\theta)g(\hat\theta)[1/\hat\theta+R(\hat\theta)]\bigg/\bigg(\int_{\underline\theta}^{\hat\theta} v(\theta)g(\theta)\mathrm{d}\theta\bigg).
\]
Because $\mathbb{E}[v(\theta)]\leq0$, we have $\int_{\underline\theta}^{\hat\theta} v(\theta)g(\theta)\mathrm{d}\theta<0$ and $1/\hat\theta+R(\hat\theta)\geq0$ by \Cref{lemma:1/theta+R}. Since $v(\hat\theta)\geq(<)0$ if and only if $\hat\theta\geq(<)\tilde\theta$, the result follows.
\end{proof}

\paragraph{Equilibrium via Mutual Best Responses.}In sum, $(\tau,\hat\theta)\in\overline{\mathbb{R}}\times(\underline\theta,\overline\theta)$ is an interior equilibrium if and only if \Cref{eq:S-BR-app,eq:R-BR} hold when $\hat\theta\in[\theta^\dagger,\overline\theta)$, and instead \Cref{eq:S-BR-app} holds and $\tau=+\infty$ when $\hat\theta\in(\underline\theta,\theta^\dagger)$. The latter case is impossible because $\tau=+\infty$ implies $\hat\theta=\overline\theta$ according to \Cref{eq:S-BR-app}, contradicting to $\hat\theta\in(\underline\theta,\theta^\dagger)$. Therefore, we can focus on $\hat\theta\in[\theta^\dagger,\overline\theta)$. 

Define
    \[
    \widehat{T}(\hat\theta;\sigma):=T(\hat\theta;\sigma)/\sigma=f^{-1}(-\sigma R(\hat\theta)),\quad\forall\hat\theta\in[\theta^\dagger,\overline\theta).
    \]
Hence, interior equilibria are characterized by fixed points of $\widehat\Theta\circ\widehat{T}$ or $\widehat\Theta\circ(-\widehat{T})$. 

For completeness, the only corner equilibrium is $(\tau,\hat\theta)$ where $\hat\theta=\overline\theta$ (i.e., all types pooling at zero effort) and $\tau$ is such that \Cref{eq:R-BR} holds. In contrast, $\hat\theta=\underline\theta$ is never possible in equilibrium by the same argument as above for the impossibility of $\hat\theta\in(\underline\theta,\theta^\dagger)$. 

\paragraph{Proof of \Cref{prop:eqm}}
\begin{proof}
    For its simplicity, we omit the existence proof for the pooling equilibrium.
    For the uniqueness result in part (a), when $\sigma/\overline\theta>f(0)$, the RHS of \Cref{eq:S-BR-app} is always smaller than the LHS for all $\hat\theta\in[\underline\theta,\overline\theta]$---that is, all types prefer zero effort regardless of the approval standard---hence $\hat\theta=\overline\theta$ and the only equilibrium is the pooling equilibrium.
    
    We now prove part (b). First, provided that for $\hat\tau<0$, $\widehat\Theta(\hat\tau)=\inf\{\theta\in[\underline\theta,\overline\theta]:\hat\tau+f^{-1}(\sigma/\theta)\geq0\}=\sigma/f(-\hat\tau)$ by \Cref{lemma:S-BR-increasing} and $1/\hat\theta>-R(\hat\theta)$ for all $\hat\theta\in(\theta^\dagger,\overline\theta)$ by \Cref{lemma:1/theta+R}, we have
    \[\widehat\Theta(-\widehat{T}(\hat\theta;\sigma);\sigma)=\inf\left\{\theta\in[\underline\theta,\overline\theta]:-f^{-1}(-\sigma R(\hat\theta))+f^{-1}(\sigma/\theta)\geq0\right\}>\hat\theta.
    \]
Hence, $\widehat\Theta\circ(-\widehat{T})$ has no fixed point in $(\theta^\dagger,\overline\theta)$. 

Then, consider $\widehat\Theta\circ\widehat{T}$. Define
    \begin{equation}
    \label{eq:tildesigma}
    \tilde\sigma:=\inf\left\{\sigma>0:\widehat\Theta(\widehat{T}(\tilde\theta;\sigma);\sigma)\geq\tilde\theta\right\}
    \quad\text{and}\quad
    \tilde\precision:=1/\tilde\sigma
    \end{equation}
    as the threshold such that $\widehat{T}$ and $\widehat\Theta$ intersect to the left of $\tilde\theta$.\footnote{It can be verified that $-\tilde{\sigma}R(\tilde\theta)\leq\tilde{\sigma}/\tilde\theta\leq f(0)$. Hence, for any $\sigma\leq\tilde\sigma$, $\widehat{T}$ indeed captures the equilibrium condition for the principal; see \Cref{fn:valid}.}
    
    Since by \Cref{lemma:T-Ushaped}, $\widehat{T}$ is U-shaped with the minimum attained at $\hat\theta=\tilde\theta$ and $\widehat{T}(\underline\theta^\dagger)=+\infty$, and by \Cref{lemma:Theta-increasing}, $\widehat\Theta$ is increasing when $\hat\tau\geq0$, as long as $\widehat\Theta(\widehat{T}(\tilde\theta;\sigma);\sigma)\leq\tilde\theta$ holds, $\widehat\Theta\circ\widehat{T}$ has a unique fixed point in $(\theta^\dagger,\tilde\theta]$, denoted by $\hat\theta_-(\sigma)\in(\theta^\dagger,\tilde\theta]$.

Notice that $\widehat\Theta(\widehat{T}(\tilde\theta;\sigma);\sigma)\leq\tilde\theta$ if and only if
\begin{equation}
\frac{\sigma}{\tilde\theta}[f^{-1}(\sigma/\tilde\theta)+f^{-1}(-\sigma R(\tilde\theta))]\leq F(f^{-1}(\sigma/\tilde\theta))-F(-f^{-1}(-\sigma R(\tilde\theta))).  \notag
\end{equation}
The LHS of this condition is converging to zero as $\sigma\to0$, while the RHS is converging to $1$. Consider the threshold $\tilde\sigma$ defined in \Cref{eq:tildesigma}. Therefore, $\tilde{\sigma}>0$ and $\widehat\Theta(\widehat{T}(\tilde\theta;\tilde{\sigma});\tilde{\sigma})=\tilde\theta$. When $\sigma\leq\tilde\sigma$, it holds that $\widehat\Theta(\widehat{T}(\tilde\theta;\sigma);\sigma)\leq\tilde\theta$. It completes the proof for part (b).
\end{proof}

\paragraph{Proof of \Cref{prop:limit}}
\begin{proof}

    Suppose not, then there exist $\epsilon$ and a subsequence $(\sigma_{k_n},(\tau_{k_n},\hat\theta_{k_n}))$ such that $\hat\theta_{k_n}>\theta^\dagger+\epsilon$ and $\hat\theta_{k_n}<\overline\theta-\epsilon$ for all $n$. For any $\hat\theta\in(\theta^\dagger,\overline\theta)$, $\lim_{\sigma\to0}T(\hat{\theta};\sigma)=\lim_{\sigma\to0}\sigma f^{-1}(-\sigma R(\hat\theta))=0$ and $\lim_{\sigma\to0}\sigma f^{-1}(\sigma/\theta)=0$ because $-\sigma R(\hat\theta)>0$ and either $f^{-1}(z)$ is bounded with bounded noise or $\lim_{p\searrow0}pf^{-1}(p)=\lim_{z\nearrow\infty}zf(z)=0$ with unbounded noise. Hence, 
    \[
    0\leq\lim_{n\to\infty}\tau_{k_n}=\lim_{n\to\infty}T(\hat\theta_{k_n};\sigma_{k_n})\leq\lim_{n\to\infty}\min\left\{T(\theta^\dagger+\epsilon;\sigma_{k_n}),T(\overline\theta-\epsilon;\sigma_{k_n})\right\}=0,
    \]
    \[
    \text{and}\quad\quad\lim_{n\to\infty}\sigma_{k_n}f^{-1}(\sigma_{k_n}/\hat\theta_{k_n})=0.
    \]
    Moreover, $\lim_{n\to\infty}\tau_{k_n}/\sigma_{k_n}=\lim_{n\to\infty}\widehat{T}(\hat\theta_{k_n};\sigma_{k_n})=\lim_{n\to\infty}f^{-1}(-\sigma_{k_n} R(\hat\theta_{k_n}))>0$.
    
    However, this leads to a violation of the indifference condition of type $\hat\theta_{k_n}$:  
    \begin{align*}
        \lim_{n\to\infty}& \left\{F(f^{-1}(\sigma_{k_n}/\hat\theta_{k_n}))-[\tau_{k_n}+\sigma_{k_n}f^{-1}(\sigma_{k_n}/\hat\theta_{k_n})]/\hat\theta_{k_n}\right\}=1\\
        &>F(0)>\lim_{n\to\infty} F(-\tau_{k_n}/\sigma_{k_n})=\lim_{n\to\infty} F_{\sigma_{k_n}}(0-\tau_{k_n}).
    \end{align*}
    As a result, either $\lim_{k\to\infty}\hat\theta_k=\theta^\dagger$ or $\lim_{k\to\infty}\hat\theta_k=\overline\theta$.
\end{proof}
The convergence of the sequence of the semi-separating equilibria $(\hat\tau_-(\sigma),\hat\theta_-(\sigma))$ is implied by the monotonicity of $\hat\tau_-(\sigma)$ and $\hat\theta_-(\sigma)$ in \Cref{lemma:increasing-indifference}, which is proved below. Since $\hat\theta_-(\sigma)$ is increasing in $\sigma$, it must converge to $\theta^\dagger$.

\subsection{Proofs from \Cref{sect:cs}}
\label{app:cs-proof}
This section contains the proofs for the comparative statics in \Cref{sect:cs}. Recall that comparative statics in $\sigma$ and $\precision$ always have opposite signs, so results in one translate directly to the other.

We first introduce a useful property of log-concave pdfs $f$ in \Cref{lemma:log-concave}. Our comparative statics results rely on the log-concavity of $f$.
\begin{lemma}
\label{lemma:log-concave}
    If $f$ is log-concave, then for any $x>x'$, 
    \[
    \frac{f'(x')}{f(x')}\geq\frac{f(x)-f(x')}{F(x)-F(x')}\geq\frac{f'(x)}{f(x)}.
    \]
    The inequalities are strict if $f'(x)<0<f'(x')$.
\end{lemma}
\begin{proof}
    For any any $x>x'$, we have
\begin{align*}
    f(x')[f(x)-f(x')]=&f(x')\int_{x'}^xf'(z)\mathrm{d}z=f(x')\int_{x'}^x\frac{f'(z)}{f(z)}f(z)\mathrm{d}z\\
    \leq&f(x')\int_{x'}^{x}\frac{f'(x')}{f(x')}f(z)\mathrm{d}z\\
    =&f'(x')\int_{x'}^{x}f(z)\mathrm{d}z=f'(x')[F(x)-F(x')],
\end{align*}
where the inequality is because $f'/f$ is decreasing by log-concavity of $f$, and we have strict inequality if either $f$ is strictly log-concave or $f'(x)<0<f'(x')$.

The other part can be derived symmetrically.
\end{proof}

\paragraph{Proof of \Cref{lemma:increasing-indifference}}

\begin{proof}
We first establish that both $\hat\theta_-(\sigma)$ and $\hat\tau_-(\sigma)$ are well-behaved.
\begin{lemma}
\label{lemma:well-behaved}
    $\hat\theta_-(\sigma)$ and $\hat\tau_-(\sigma)$ are twice continuously differentiable in $\sigma\in(0,\tilde\sigma]$.
\end{lemma}
\begin{proof}
First, $\widehat{T}\in\mathcal{C}^2((\theta^\dagger,\overline\theta)\times(0,\tilde\sigma])$ since both $f^{-1}$ (except at $f(0)$) and $R$ are twice continuously differentiable. 
Second, since $f\in\mathcal{C}^2(\mathbb{R})$ and $\widehat\Theta$ is the implicit function defined by \Cref{eq:S-BR-app}, $\widehat\Theta$ is $\mathcal{C}^2$ over $\big\{(\hat\tau,\sigma)\in\overline{\mathbb{R}}\times(0,\tilde\sigma]:\widehat\Theta(\hat\tau;\sigma)\in(\underline\theta,\overline\theta)\big\}$ (an open set in $\mathbb{R}^2$ with compact closure, by monotonicity of $\widehat\Theta$). Therefore, as the unique fixed point of 
$\widehat\Theta\circ\widehat{T}$ in $(\theta^\dagger,\tilde\theta]$, $\hat\theta_-$ and hence $\hat\tau_-$ are continuous in $\sigma$. 

To show that $\hat\theta_-$ is continuously differentiable, it suffices by the Implicit Function Theorem to verify that the partial derivative of $P(\hat\theta,\sigma):=\widehat\Theta(\widehat{T}(\hat\theta;\sigma);\sigma)-\hat\theta$ with respect to $\hat\theta$ is nonzero at $\hat\theta=\hat\theta_-(\sigma)$:
\begin{equation}
\label{eq:dP/dtheta}
\frac{\partial P}{\partial\hat\theta}\big|_{\hat\theta=\hat\theta_-(\sigma)}=\left[\widehat\Theta'(\widehat{T}(\hat\theta;\sigma))\widehat{T}'(\hat\theta)-1\right]\big|_{\hat\theta=\hat\theta_-(\sigma)}<0,
\end{equation}
since $\widehat{T}'(\hat\theta)\leq0$ for $\hat\theta\leq\tilde\theta$ by \Cref{lemma:R-Ushaped} and $\widehat\Theta'\geq0$ by \Cref{lemma:S-BR-increasing} provided that $\widehat{T}\geq0$. As a result, $\hat\theta_-$ and $\hat\tau_-$ are continuously differentiable in $\sigma\in(0,\tilde\sigma]$. Since $\widehat{T}$ and $\hat\theta$ are twice continuously differentiable, so are $\hat\theta_-$ and $\hat\tau_-$.
\end{proof}

In particular, with $P(\hat\theta,\sigma)=\widehat\Theta(\widehat{T}(\hat\theta;\sigma);\sigma)-\hat\theta$,
\begin{equation}
\label{eq:dhattheta/dsigma}
\frac{\mathrm{d}\hat\theta_-}{\mathrm{d}\sigma}=-\frac{\partial P/\partial \sigma}{\partial P/\partial \hat\theta}\Big|_{\hat\theta=\hat\theta_-(\sigma)}=\frac{\frac{\partial\widehat\Theta}{\partial\sigma}+\widehat\Theta'\frac{\partial\widehat{T}}{\partial\sigma}}{1-\widehat\Theta'\widehat{T}'}\Big|_{\hat\theta=\hat\theta_-(\sigma)}
\end{equation}
where $\widehat\Theta'$ and $\widehat{T}'$ denote $\partial\widehat\Theta/\partial\hat\tau$ and $\partial\widehat{T}/\partial\hat\theta$. Given that $(\partial P/\partial \hat\theta)|_{\hat\theta=\hat\theta_-(\sigma)}<0$ in (\ref{eq:dP/dtheta}), to show $\hat\theta_-(\sigma)$ is strictly increasing in $\sigma\in(0,\tilde\sigma]$, it thus suffices to prove
\[
\left[\frac{\partial\widehat\Theta}{\partial\sigma}+\widehat\Theta'\frac{\partial\widehat{T}}{\partial\sigma}\right]\Big|_{\hat\theta=\hat\theta_-(\sigma)}>0.
\]
In detail, we have $\partial\widehat\Theta/\partial\sigma|_{\hat\theta=\widehat\Theta(\hat\tau)}=\hat\theta/\sigma$,
\[
\widehat\Theta'\big|_{\hat\theta=\widehat\Theta(\hat\tau)}=\frac{\sigma/\hat\theta-f(-\hat\tau)}{[\hat\tau+f^{-1}(\sigma/\hat\theta)]\sigma/\hat\theta^2}=\hat\theta\cdot\frac{f(f^{-1}(\sigma/\hat\theta))-f(-\hat\tau)}{F(f^{-1}(\sigma/\hat\theta))-F(-\hat\tau)}
\]
\[
\text{and}\quad\frac{\partial\widehat{T}}{\partial\sigma}\Big|_{\hat\tau=\widehat{T}(\hat\theta)}=\frac{- R(\hat\theta)}{f'(f^{-1}(-\sigma R(\hat\theta)))}\Big|_{\hat\tau=\widehat{T}(\hat\theta)}=\frac{(1/\sigma)f(-\hat\tau)}{-f'(-\hat\tau)},
\]
where the last equality is by symmetry of $f$ and $\hat\tau=\widehat{T}(\hat\theta)=f^{-1}(-\sigma R(\hat\theta))$.
Hence,
\begin{equation}
\label{eq:dhattheta/dsigma-num}
\left[\frac{\partial\widehat\Theta}{\partial\sigma}+\widehat\Theta'\frac{\partial\widehat{T}}{\partial\sigma}\right]\Big|_{\hat\theta=\hat\theta_-(\sigma)}=\frac{\hat\theta}{\sigma}\left[1-\frac{f(f^{-1}(\sigma/\hat\theta))-f(-\hat\tau)}{F(f^{-1}(\sigma/\hat\theta))-F(-\hat\tau)}\cdot\frac{f(-\hat\tau)}{f'(-\hat\tau)}\right]\Big|_{\hat\theta=\hat\theta_-(\sigma),\hat\tau=\widehat{T}(\hat\theta)}>0,
\end{equation}
by \Cref{lemma:log-concave} and that $f^{-1}(\sigma/\hat\theta)>0>-\hat\tau$ (with $f'(-\hat\tau)>0>f'(f^{-1}(\sigma/\hat\theta))$).

As for $\hat\tau_-(\sigma)$, since $\hat\tau_-(\sigma)=\widehat{T}(\hat\theta_-(\sigma);\sigma)$, we have
\begin{equation}
\label{eq:dhattau/dsigma}
    \frac{\mathrm{d}\hat\tau_-}{\mathrm{d}\sigma}=\frac{-R(\hat\theta_-)-\sigma R'(\hat\theta_-)\frac{\mathrm{d}\hat\theta_-}{\mathrm{d}\sigma}}{f'(f^{-1}(-\sigma R(\hat\theta_-)))}<0\quad \text{given that }\hat\theta_-(\sigma)\in(\theta^\dagger,\tilde\theta].
\end{equation}
    This completes the proof.
\end{proof}

\paragraph{Proof of \Cref{prop:non-monotone}}
\begin{proof}
By symmetry of $f$, the principal's payoff $V$ can be rewritten as
\begin{align*}
V(\sigma)=\int_{\underline\theta}^{\hat\theta_-(\sigma)} v(\theta)g(\theta)\mathrm{d}\theta\cdot F(-\hat\tau_-(\sigma))+\int_{\hat\theta_-(\sigma)}^{\overline\theta} v(\theta)g(\theta)F(f^{-1}(\sigma/\theta))\mathrm{d}\theta.
\end{align*}
    By \Cref{lemma:well-behaved}, $V(\sigma)$ should also be continuously differentiable. Hence,
    \begin{align*}
        V'(\sigma)=&-\int_{\underline\theta}^{\hat\theta_-(\sigma)} v(\theta)g(\theta)\mathrm{d}\theta\cdot f(0-\hat\tau_-(\sigma))\frac{\mathrm{d}\hat\tau_-}{\mathrm{d}\sigma}+\int_{\hat\theta_-(\sigma)}^{\overline\theta} v(\theta)g(\theta)/\theta\frac{\sigma/\theta}{f'(f^{-1}(\sigma/\theta))}\mathrm{d}\theta\\
        &+v(\hat\theta_-(\sigma))g(\hat\theta_-(\sigma))\left[F(0-\hat\tau_-(\sigma))-F(f^{-1}(\sigma/\hat\theta_-(\sigma)))\right]\frac{\mathrm{d}\hat\theta_-}{\mathrm{d}\sigma}\\
        =&:A_1(\sigma)+A_2(\sigma)+A_3(\sigma).
    \end{align*}
    In particular, $A_1,A_2$ and $A_3$ are all continuous in $\sigma$.
    
    For part (a), notice that $A_1<0$ since $\frac{\mathrm{d}\hat\tau_-}{\mathrm{d}\sigma}<0$ by \Cref{lemma:increasing-indifference}. And when $\sigma=\tilde\sigma$ and thus $\hat\theta_-(\sigma)=\tilde\theta$, $A_2<0$ and $A_3=0$. As a result, $V'(\tilde\sigma)<0$.

    For part (b), first notice that when $\sigma<\tilde\sigma$, $A_3>0$ since $\frac{\mathrm{d}\hat\theta_-}{\mathrm{d}\sigma}>0$ by \Cref{lemma:increasing-indifference} and $v(\hat\theta_-(\sigma))<0$ provided that $\hat\theta_-(\sigma)<\tilde\theta$. Next, we establish two claims: (1) $\lim_{\sigma\to0}A_1/A_3=0$ and (2) $\lim_{\sigma\to0}A_2/A_3\geq0$. With (1) and (2), given $A_3>0$ and continuity of $A_3$, we conclude that $\underline\sigma>0$ exists such that $V'(\sigma)>0$ for $\sigma\in(0,\underline\sigma)$.

    It is useful to first establish the below lemma that $\hat\theta_-(\sigma)$ converges to $\theta^\dagger$ at a faster rate than $\frac{\mathrm{d}\hat\theta_-}{\mathrm{d}\sigma}$ converges to 0. (Recall $\lim_{\sigma\to0}\hat\theta_-(\sigma)=\theta^\dagger$ by \Cref{prop:limit}.)
    \begin{lemma}
        \label{lm:dtheta/dsigma-order}
        It holds that 
        $
        \lim_{\sigma\to0}\frac{R(\hat\theta_-(\sigma))}{\mathrm{d}\hat\theta_-/\mathrm{d}\sigma}=0
        $, and thus
        $
        \lim_{\sigma\to0}\frac{\hat\theta_-(\sigma)-\theta^\dagger}{\mathrm{d}\hat\theta_-/\mathrm{d}\sigma}=0.
        $
    \end{lemma}
    \begin{proof}
    According to the proof of \Cref{lemma:increasing-indifference}, in particular, \Cref{eq:dhattheta/dsigma,eq:dhattheta/dsigma-num},
        \begin{align*}
\frac{R(\hat\theta_-(\sigma))}{\mathrm{d}\hat\theta_-/\mathrm{d}\sigma}&=\frac{-f(-\hat\tau)}{\sigma}\frac{1-\frac{\sigma/\hat\theta-f(-\hat\tau)}{[\hat\tau+f^{-1}(\sigma/\hat\theta)]\sigma/\hat\theta^2}\frac{\sigma R'(\hat\theta)}{f'(-\hat\tau)}}{\frac{\hat\theta}{\sigma}-\frac{\hat\theta}{\sigma}\frac{\sigma/\hat\theta-f(-\hat\tau)}{[\hat\tau+f^{-1}(\sigma/\hat\theta)]\sigma/\hat\theta}\cdot\frac{f(-\hat\tau)}{f'(-\hat\tau)}}\\
&=\frac{(\sigma/\hat\theta^2)[\hat\tau+f^{-1}(\sigma/\hat\theta)]f'(-\hat\tau)-[\sigma/\hat\theta-f(-\hat\tau)]\sigma R'(\hat\theta)}{-\frac{\sigma}{\hat\theta}[\hat\tau+f^{-1}(\sigma/\hat\theta)] \frac{f'(-\hat\tau)}{f(-\hat\tau)}+[\sigma/\hat\theta-f(-\hat\tau)]}.
\end{align*}
The second term in the numerator clearly converges to zero as $\sigma\to0$; the first term does as well, since $\sigma\hat\tau_-(\sigma)=\tau_-(\sigma)\to\theta^\dagger$, $(\sigma/\hat\theta_-(\sigma))f^{-1}(\sigma/\hat\theta_-(\sigma))\to0$, and $f'(-\hat\tau_-(\sigma))\to0$. As for the denominator, the second term clearly converges to zero, while the first term is bounded away from zero: $\frac{\sigma}{\hat\theta}[\hat\tau+f^{-1}(\sigma/\hat\theta)]\to\frac{\theta^\dagger}{\theta^\dagger}+0=1$ and $f'(-\hat\tau)/f(-\hat\tau)$ is increasing as $\sigma\to0$ and hence bounded below by some $c>0$ (e.g., $c=f'(-1)/f(-1)$).

As a result, $\lim_{\sigma\to0}\frac{R(\hat\theta_-(\sigma))}{\mathrm{d}\hat\theta_-/\mathrm{d}\sigma}=0$. Moreover, because $R'(\theta^\dagger)\ne0$,  $R(\hat\theta_-(\sigma))$ is of the same order as $\hat\theta_-(\sigma)-\theta^\dagger$, therefore we also have $\lim_{\sigma\to0}\frac{\hat\theta_-(\sigma)-\theta^\dagger}{\mathrm{d}\hat\theta_-/\mathrm{d}\sigma}=0$.
    \end{proof}

    \begin{claim}
    \label{claim:a1}
    $\lim_{\sigma\to0}A_1(\sigma)/A_3(\sigma)=0$.
    \end{claim}
    \begin{proof}
    It suffices to show $\lim_{\sigma\to0}f(0-\hat\tau_-(\sigma))\frac{\mathrm{d}\hat\tau_-/\mathrm{d}\sigma}{\mathrm{d}\hat\theta_-/\mathrm{d}\sigma}=0$ since other terms are all bounded from both zero and infinity. According to \Cref{eq:dhattau/dsigma}, 
    \[
    \lim_{\sigma\to0}f(0-\hat\tau_-(\sigma))\frac{\mathrm{d}\hat\tau_-/\mathrm{d}\sigma}{\mathrm{d}\hat\theta_-/\mathrm{d}\sigma}=\lim_{\sigma\to0}\frac{f(-\hat\tau_-(\sigma))}{f'(-\hat\tau_-(\sigma))}\left[\frac{R(\hat\theta_-(\sigma))}{\mathrm{d}\hat\theta_-/\mathrm{d}\sigma}+\sigma R'(\hat\theta_-(\sigma))\right],
    \]
    where $\frac{f(-\hat\tau_-(\sigma))}{f'(-\hat\tau_-(\sigma))}$ is bounded away from infinity when $\sigma$ is small since this ratio is positive and increasing in $\sigma$, and $R'(\hat\theta_-(\sigma))$ is also bounded. Hence it is sufficient to show $\lim_{\sigma\to0}R(\hat\theta_-(\sigma))/\frac{\mathrm{d}\hat\theta_-}{\mathrm{d}\sigma}=0$, which is completed by \Cref{lm:dtheta/dsigma-order}. 
    \end{proof}
    \begin{claim}
    \label{claim:A2}
    $\lim_{\sigma\to0}A_2(\sigma)/A_3(\sigma)\geq0$.
    \end{claim}
    \begin{proof}    
    We can rewrite $A_2$ as
    \begin{align*}
    A_2(\sigma)&=\int_{\theta^\dagger}^{\overline\theta} v(\theta)g(\theta)/\theta\frac{f(f^{-1}(\sigma/\theta))}{f'(f^{-1}(\sigma/\theta))}\mathrm{d}\theta-\int_{\theta^\dagger}^{\hat\theta_-(\sigma)} v(\theta)g(\theta)/\theta\frac{f(f^{-1}(\sigma/\theta))}{f'(f^{-1}(\sigma/\theta))}\mathrm{d}\theta\\
    &=:\overline{A}_2(\sigma)-\Delta(\sigma).
    \end{align*}
    Note that $\overline{A}_2(\sigma)\geq0$ for all $\sigma<\theta^\dagger f(0)$, since $\int_{\theta^\dagger}^{\overline\theta} v(\theta)g(\theta)/\theta\,\mathrm{d}\theta=0$ and $f'/f<0$ is decreasing by log-concavity of $f$. Hence, to prove the statement, it suffices to show $\lim_{\sigma\to0}\Delta(\sigma)/A_3(\sigma)=0$, or equivalently $\lim_{\sigma\to0}\Delta(\sigma)/\frac{\mathrm{d}\hat\theta_-}{\mathrm{d}\sigma}=0$. Since $\frac{f(f^{-1}(\sigma/\theta))}{f'(f^{-1}(\sigma/\theta))}$ is bounded for small $\sigma$, $\Delta(\sigma)$ has smaller order than $\hat\theta_-(\sigma)-\theta^\dagger$. By \Cref{lm:dtheta/dsigma-order}, we conclude $\lim_{\sigma\to0}\Delta(\sigma)/\frac{\mathrm{d}\hat\theta_-}{\mathrm{d}\sigma}=0$ and $\lim_{\sigma\to0}A_2(\sigma)/A_3(\sigma)\geq0$.
    \end{proof}
    These two claims complete the proof.
\end{proof}

\paragraph{Proof of \Cref{prop:approval}: Part (a)}
\begin{proof}
We have
\begin{align*}
    AR'(\sigma)
    &=-G(\hat\theta_-(\sigma))f(0-\hat\tau_-(\sigma))\frac{\mathrm{d}\hat\tau_-}{\mathrm{d}\sigma}
    +\int_{\hat\theta_-(\sigma)}^{\overline\theta}\frac{g(\theta)}{\theta}
    \cdot\frac{\sigma/\theta}{f'(f^{-1}(\sigma/\theta))}\mathrm{d}\theta\\
    +&g(\hat\theta_-(\sigma))\left[F(0-\hat\tau_-(\sigma))-F(f^{-1}(\sigma/\hat\theta_-(\sigma)))\right]
    \frac{\mathrm{d}\hat\theta_-}{\mathrm{d}\sigma}=: C_1(\sigma)+C_2(\sigma)+C_3(\sigma).
\end{align*}
Note that $C_2,C_3<0$ and $C_1>0$ for all $\sigma$. Since $\lim_{\sigma\to0}f(0-\hat\tau_-(\sigma))\frac{\mathrm{d}\hat\tau_-/\mathrm{d}\sigma}{\mathrm{d}\hat\theta_-/\mathrm{d}\sigma}=0$, we have $\lim_{\sigma\to0}C_1(\sigma)/C_3(\sigma)=0$. Hence by continuity, there exists $\underline\sigma'>0$ such that $AR'(\sigma)<0$ for all $\sigma\in(0,\underline\sigma')$.
\end{proof}

\paragraph{Proof of \Cref{prop:approval}: Part (b)}
The proof of part (b) is much more involved. We record the additional assumption here.
\begin{assumption}
\label{as:limit}
    $\lim_{z\to\infty}\frac{f(z)f''(z)}{[f'(z)]^2}=1$.
\end{assumption}
The proof idea is to show there exists $b(\sigma)>0$ such that, as $\sigma\to0$,
\[
U'(\sigma)
=
-b(\sigma)\int_{\hat\theta}^{\bar\theta}\frac{g(\theta)}{\theta}
\log\frac{\theta}{\theta^\dagger}\,\mathrm{d}\theta
+
o\bigl(b(\sigma)\bigr).
\]
For this sake, define
\[
\phi(x):=-\log f(x), \quad \text{for }x>0.
\]
By symmetry and log-concavity, $\phi$ is increasing and convex on $(0,\infty)$. Moreover, \Cref{as:limit} is equivalent to
\begin{equation}
\label{as:limit-2}
\frac{\phi''(x)}{(\phi'(x))^2}\to0
\qquad\text{as }x\to\infty.\tag{Assumption 3'}
\end{equation}
Define
\[
b(p):=\frac1{\phi'(f^{-1}(p))}=-\frac{f(f^{-1}(p))}{f'(f^{-1}(p))}=-\frac{p}{f'(f^{-1}(p))},
\quad \text{for }p>0.
\]
Hence $b(p)>0$, $b$ is increasing in $p$, and $b(p)=O(1)$ as $p\to0$.

It is useful to establish a series of useful asymptotics related to $b(p)$. We relegate the proofs to \Cref{sect:asymptotics}.

\begin{lemma}
\label{lm:b-linear}
    For any fixed $c>0$,
\[
b(cp)=b(p)+o\bigl(b(p)\bigr),
\quad\text{as }p\to0.
\]
\end{lemma}

\begin{lemma}
\label{lm:f-inverse}
    For any compact $C\subset (0,\infty)$,
\[
f^{-1}(cp)-f^{-1}(p)=-b(p)\log c+o\bigl(b(p)\bigr)
\quad\text{uniformly in $c\in C$, as }p\to0.
\]
\end{lemma}

\begin{lemma}
\label{lm:tail-bounds}
For any fixed $t>0$,
\[
f(t/p)=o\bigl(p^2b(p)\bigr)
\qquad\text{and}\qquad
f'(t/p)=o\bigl(p^3 b(p)\bigr),\quad\text{as }p\to0.
\]
\end{lemma}

Given these asymptotic facts, we can establish that when $\sigma$ is small, the effects of $\sigma$ on both $\hat\theta_-(\sigma)$ and $F(-\hat\tau_-(\sigma))$ are of smaller order than $b(\sigma)$.
\begin{lemma}
\label{lm:dthetahat-small}
As $\sigma\to0$,
\[
\frac{\mathrm{d}\hat\theta_-}{\mathrm{d}\sigma}=o(b(\sigma))
\quad\text{and}\quad
\frac{\mathrm{d}}{\mathrm{d}\sigma}F(-\hat\tau_-(\sigma))=o(b(\sigma)).
\]
\end{lemma}
\begin{proof}
We abuse notation and write $\hat\theta:=\hat\theta_-(\sigma),\hat\tau:=\hat\tau_-(\sigma)$ and $\tau:=\tau_-(\sigma)$. Recall that, according to \Cref{eq:dhattheta/dsigma},
\begin{align*}
\frac{\mathrm{d}\hat\theta_-}{\mathrm{d}\sigma}=&\frac{\frac{\hat\theta}{\sigma}-\frac{\hat\theta}{\sigma}\frac{\sigma/\hat\theta-f(-\hat\tau)}{[\hat\tau+f^{-1}(\sigma/\hat\theta)]\sigma/\hat\theta}\cdot\frac{f(-\hat\tau)}{f'(-\hat\tau)}}{1-\frac{\sigma/\hat\theta-f(-\hat\tau)}{[\hat\tau+f^{-1}(\sigma/\hat\theta)]\sigma/\hat\theta^2}\frac{\sigma R'(\hat\theta)}{f'(-\hat\tau)}}\\
=&\frac{\frac{\sigma}{\hat\theta}[\hat\tau+f^{-1}(\sigma/\hat\theta)]f'(-\hat\tau)-[\sigma/\hat\theta-f(-\hat\tau)]f(-\hat\tau)}{(\sigma/\hat\theta)^2[\hat\tau+f^{-1}(\sigma/\hat\theta)]f'(-\hat\tau)-[\sigma/\hat\theta-f(-\hat\tau)]\sigma^2 R'(\hat\theta)}
\end{align*}
We now estimate the numerator and denominator. Since $\hat\tau_-(\sigma)=\tau_-(\sigma)/\sigma$ and $\tau_-(\sigma)\to\theta^\dagger$, by \Cref{lm:tail-bounds} and symmetry of $f$,
\[
f(-\hat\tau)=o\bigl(\sigma^2b(\sigma)\bigr)
\qquad\text{and}\qquad
f'(-\hat\tau)=o\bigl(\sigma^3 b(\sigma)\bigr).
\]
Given that $\sigma\hat\tau_-(\sigma)\to\theta^\dagger$ and $\lim_{p\searrow0}p f^{-1}(p)=0$, the numerator is thus
\[
o\bigl(\sigma^3b(\sigma)\bigr)+o\bigl(\sigma^3b(\sigma)\bigr)=o\bigl(\sigma^3b(\sigma)\bigr).
\]
Similarly, since $R'(\theta^\dagger)\ne0$ and $b(\sigma)=O(1)$, the denominator is
\[
o\bigl(\sigma^4b(\sigma)\bigr)-O(\sigma^3)\cdot R'(\theta^\dagger)=O(\sigma^3).
\]
Combining the above estimates, we conclude that
\[
\mathrm{d}\hat\theta_-/\mathrm{d}\sigma=o\bigl(b(\sigma)\bigr).
\]
Moreover, recall that $\lim_{\sigma\to0}f(0-\hat\tau_-(\sigma))\frac{\mathrm{d}\hat\tau_-/\mathrm{d}\sigma}{\mathrm{d}\hat\theta_-/\mathrm{d}\sigma}=0$, as shown in the proof of \Cref{prop:non-monotone}. As a result, $\frac{\mathrm{d}}{\mathrm{d}\sigma}F(-\hat\tau_-(\sigma))=o(\mathrm{d}\hat\theta_-/\mathrm{d}\sigma)=o(b(\sigma))$.
\end{proof}

Now we are ready to prove part (b) of \Cref{prop:approval}.
\begin{proof}[Proof of \Cref{prop:approval}: Part (b)]
We abuse notation and write $\hat\theta:=\hat\theta_-(\sigma),\hat\tau:=\hat\tau_-(\sigma)$ and $\tau:=\tau_-(\sigma)$. By the indifference condition,
\begin{align*}
    \tau 
&=\hat\theta F(f^{-1}(\sigma/\hat\theta))
-\sigma f^{-1}(\sigma/\hat\theta)-\hat\theta F(-\hat\tau)\\
&=\hat\theta-\sigma\left[f^{-1}(\sigma/\hat\theta)+\frac{1-F(f^{-1}(\sigma/\hat\theta))}{f(f^{-1}(\sigma/\hat\theta))}\right]-\hat\theta F(-\hat\tau)\\
&=\hat\theta-\sigma\Psi(f^{-1}(\sigma/\hat\theta))-\hat\theta F(-\hat\tau),
\end{align*}
where
\[
\Psi(x):=x+\frac{1-F(x)}{f(x)}.
\]
Define
\[
r(\sigma)
:=
(\hat\theta-\theta^\dagger)-\hat\theta F(-\hat\tau)
-\sigma\left[\Psi(f^{-1}(\sigma/\hat\theta))-\Psi(f^{-1}(\sigma/\theta^\dagger))\right].
\]
Then
\[
\tau
=
\theta^\dagger-\sigma\Psi(f^{-1}(\sigma/\theta^\dagger))+r(\sigma).
\]
For $\theta\ge\hat\theta$, agent type $\theta$'s equilibrium payoff can thus be written as
\begin{align*}
    u(\theta,\sigma)
:=&
F(f^{-1}(\sigma/\theta))-\frac{\tau+\sigma f^{-1}(\sigma/\theta)}{\theta}=1-\frac{\theta^\dagger}{\theta}
-\frac{\sigma}{\theta}\left[\Psi(f^{-1}(\sigma/\theta))-\Psi(f^{-1}(\sigma/\theta^\dagger))\right]
-\frac{r(\sigma)}{\theta}\\
=&1-\frac{\theta^\dagger}{\theta}
-\frac{\sigma}{\theta}D_\Psi(\sigma,\theta)
-\frac{r(\sigma)}{\theta},
\end{align*}
where
\[
D_\Psi(\sigma,\theta):=\Psi(f^{-1}(\sigma/\theta))-\Psi(f^{-1}(\sigma/\theta^\dagger)).
\]
\begin{claim}
\label{claim:diff}
    As $\sigma$ goes to $0$,
    \[
    \frac{\partial}{\partial\sigma}\left[\sigma D_\Psi(\sigma,\theta)\right]=b(\sigma)\log(\theta/\theta^\dagger)+
    o\bigl(b(\sigma)\bigr),\text{ uniformly for } \theta\in[\theta^\dagger,\bar\theta].
    \]
\end{claim}
\begin{proof}
Notice that
\[
\Psi'(x)=1+\frac{-[f(x)]^2-[1-F(x)]f'(x)}{[f(x)]^2}=\frac{-[1-F(x)]f'(x)}{[f(x)]^2}=\frac{1-F(x)}{f(x)/\phi'(x)}.
\]
Since $1-F(x)\to0$ and $f(x)/\phi'(x)\to0$, by l'Hospital's rule and \ref{as:limit-2},
\[
\lim_{x\to\infty}\frac{1-F(x)}{f(x)/\phi'(x)}=\lim_{x\to\infty}\frac{-f(x)}{\frac{f'(x)\phi'(x)-f(x)\phi''(x)}{[\phi'(x)]^2}}=\lim_{x\to\infty}\frac{1}{1+\frac{\phi''(x)}{[\phi'(x)]^2}}=1.
\]
Therefore,
\[
\Psi'(x)=1+o(1)
\qquad\text{as }x\to\infty.
\]
Hence,
\begin{align*}
    D_\Psi(\sigma,\theta)=\Psi(f^{-1}(\sigma/\theta))-\Psi(f^{-1}(\sigma/\theta^\dagger))
=&\Psi'(\xi)[f^{-1}(\sigma/\theta)-f^{-1}(\sigma/\theta^\dagger)]\\
=&
b(\sigma)\log(\theta/\theta^\dagger)
+
o\bigl(b(\sigma)\bigr),
\end{align*}
where the last equality is by \Cref{lm:f-inverse}. Moreover,
\[
\frac{\partial}{\partial\sigma}f^{-1}(\sigma/\theta)=\frac{1/\theta}{f'(f^{-1}(\sigma/\theta))}=\frac{-b(\sigma/\theta)}{\sigma},
\]
so by \Cref{lm:b-linear},
\[
\sigma\frac{\partial}{\partial\sigma}D_\Psi(\sigma,\theta)
=
-\Psi'(f^{-1}(\sigma/\theta))b(\sigma/\theta)+\Psi'(f^{-1}(\sigma/\theta^\dagger))b(\sigma/\theta^\dagger)
=
o\bigl(b(\sigma)\bigr)
\]
uniformly in $\theta$. As a result, uniformly for $\theta\in[\theta^\dagger,\bar\theta]$,
\[
    \frac{\partial}{\partial\sigma}\left[\sigma D_\Psi(\sigma,\theta)\right]=D_\Psi(\sigma,\theta)+\sigma\frac{\partial}{\partial\sigma}D_\Psi(\sigma,\theta)=b(\sigma)\log(\theta/\theta^\dagger)+
    o\bigl(b(\sigma)\bigr).
    \]
This completes the proof.
\end{proof}
Therefore, uniformly for $\theta\in[\theta^\dagger,\bar\theta]$,
\[
\frac{\partial u(\theta,\sigma)}{\partial\sigma}
=
-\frac{b(\sigma)}{\theta}\log\frac{\theta}{\theta^\dagger}
+
o\bigl(b(\sigma)\bigr)+\frac{r'(\sigma)}{\theta}.
\]
Moreover, we also have $r'(\sigma)=o(b(\sigma))$ by Claim \ref{claim:diff} and \Cref{lm:dthetahat-small}. As a result,
\[
\frac{\partial u(\theta,\sigma)}{\partial\sigma}
=
-\frac{b(\sigma)}{\theta}\log\frac{\theta}{\theta^\dagger}
+
o\bigl(b(\sigma)\bigr),\quad\text{uniformly for } \theta\in[\theta^\dagger,\bar\theta].
\]

Recall that
\[
U(\sigma)
=
G(\hat\theta)F(-\hat\tau)
+
\int_{\hat\theta}^{\bar\theta}g(\theta)u(\theta,\sigma)\,\mathrm{d}\theta.
\]
Hence,
\begin{align*}
U'(\sigma)=&
g(\hat\theta)\frac{\mathrm{d}\hat\theta}{\mathrm{d}\sigma}F(-\hat\tau)
+
G(\hat\theta)\frac{\mathrm{d}}{\mathrm{d}\sigma}F(-\hat\tau)-g(\hat\theta)u(\hat\theta,\sigma)\frac{\mathrm{d}\hat\theta}{\mathrm{d}\sigma}
+
\int_{\hat\theta}^{\bar\theta}g(\theta)\frac{\partial u(\theta,\sigma)}{\partial\sigma}\,\mathrm{d}\theta\\
=&
G(\hat\theta)\frac{\mathrm{d}}{\mathrm{d}\sigma}F(-\hat\tau)
+
\int_{\hat\theta}^{\bar\theta}g(\theta)\frac{\partial u(\theta,\sigma)}{\partial\sigma}\,\mathrm{d}\theta.
\end{align*}
By the previous estimate for $\partial u(\theta,\sigma)/\partial\sigma$, we have
\[
U'(\sigma)
=
-b(\sigma)\int_{\hat\theta}^{\bar\theta}\frac{g(\theta)}{\theta}
\log\frac{\theta}{\theta^\dagger}\,\mathrm{d}\theta
+
o\bigl(b(\sigma)\bigr).
\]
Finally, since $\hat\theta\to\theta^\dagger$, by dominated convergence theorem,
\[
U'(\sigma)=-b(\sigma)\int_{\theta^\dagger}^{\bar\theta}\frac{g(\theta)}{\theta}
\log\frac{\theta}{\theta^\dagger}\,\mathrm{d}\theta+o\bigl(b(\sigma)\bigr).
\]
Since $b(\sigma)>0$ and $(g(\theta)/\theta)
\log(\theta/\theta^\dagger)>0$, there exists $\underline\sigma'>0$ such that
\[
U'(\sigma)<0
\qquad\text{for all }\sigma\in(0,\underline\sigma').
\]
This completes the proof.
\end{proof}

\paragraph{Proof of \Cref{prop:error}}
\begin{proof}
    The part for $\alpha(\sigma)$ is straightforward. As for $\beta(\sigma)$, we have
    \begin{align*}
        \beta'(\sigma)=&-G(\hat\theta_-(\sigma)) f(-\hat\tau_-(\sigma))\frac{\mathrm{d}\hat\tau_-}{\mathrm{d}\sigma}+\int_{\hat\theta_-(\sigma)}^{\tilde\theta} g(\theta)/\theta\frac{\sigma/\theta}{f'(f^{-1}(\sigma/\theta))}\mathrm{d}\theta\\
        &+g(\hat\theta_-(\sigma))\left[F(-\hat\tau_-(\sigma))-F(f^{-1}(\sigma/\hat{\theta}_-(\sigma)))\right]\frac{\mathrm{d}\hat\theta_-}{\mathrm{d}\sigma}\\
        =&:B_1(\sigma)+B_2(\sigma)+B_3(\sigma).
    \end{align*}
    By \Cref{lemma:increasing-indifference}, for $\sigma\in(0,\tilde\sigma)$, $B_1>0$ while $B_2,B_3<0$. By the proof of \Cref{prop:non-monotone}, $\lim_{\sigma\to0}f(-\hat\tau_-(\sigma))\frac{\mathrm{d}\hat\tau_-}{\mathrm{d}\sigma}/\frac{\mathrm{d}\hat\theta_-}{\mathrm{d}\sigma}=0$, so $\lim_{\sigma\to0}B_1/B_3=0$. Hence, $\lim_{\sigma\to0}\beta'(\sigma)<0$.
\end{proof}

\paragraph{Proof of \Cref{prop:accuracy}}
\begin{proof}
Note that
\[
H_\sigma(H_\sigma^{-1}(p|\theta)|\theta')=F\left(F^{-1}(p)+\frac{e^*(\theta;\sigma)-e^*(\theta';\sigma)}{\sigma}\right).
\]
For any two noise levels $0<\sigma<\sigma'<\tilde{\sigma}$: 
When $\theta<\hat\theta_-(\sigma)\leq\tilde\theta<\theta'$: since 
    \[
    \frac{e^*(\theta;\sigma)-e^*(\theta';\sigma)}{\sigma}=-\hat\tau_-(\sigma)-f^{-1}(\sigma/\theta')<-\hat\tau_-(\sigma')-f^{-1}(\sigma'/\theta')=\frac{e^*(\theta;\sigma')-e^*(\theta';\sigma')}{\sigma'},
    \]
more accurate information is transmitted under $\sigma$ compared to $\sigma'$.

When $\hat\theta_-(\sigma')<\theta<\tilde\theta<\theta'$: we have
    \[
    \frac{e^*(\theta;\sigma)-e^*(\theta';\sigma)}{\sigma}=f^{-1}(\sigma/\theta)-f^{-1}(\sigma/\theta')>f^{-1}(\sigma'/\theta)-f^{-1}(\sigma'/\theta')=\frac{e^*(\theta;\sigma')-e^*(\theta';\sigma')}{\sigma'},
    \]
where the inequality is due to $f'/f$ being decreasing. Hence, less accurate information is transmitted under $\sigma$ than under $\sigma'$.

When $\hat\theta_-(\sigma)<\theta<\hat\theta_-(\sigma')<\tilde\theta<\theta'$: since
    \[
    \frac{e^*(\theta;\sigma)-e^*(\theta';\sigma)}{\sigma}=f^{-1}(\sigma/\theta)-f^{-1}(\sigma/\theta')>-\hat\tau_-(\sigma')-f^{-1}(\sigma'/\theta')=\frac{e^*(\theta;\sigma')-e^*(\theta';\sigma')}{\sigma'},
    \]
    less information is transmitted under $\sigma$ than under $\sigma'$.
\end{proof}

\subsection{Proofs from \Cref{sect:discussion}}
\paragraph{Proof of \Cref{prop:absorbing}}
\begin{proof}
To see why the principal will choose $\precision_0$, recall that by the SCP of the agent's payoff, $e(\theta)$ must be increasing in $\theta$ in equilibrium. For any increasing $e(\theta)$, the induced information structure $k_\precision(s|\theta):=f_\precision(s-e(\theta))$ satisfies MLRP and is ranked with respect to the accuracy order: the larger $\precision$ is, the more accurate $k_\precision$ is. Then by the SCP of the principal's payoff and Theorem 1 in \citet{persico2000information}, the principal always (weakly) prefers larger $\precision$. 
\end{proof}

\paragraph{Proof of \Cref{prop:commit}}
\begin{proof}
Note that the first-order condition for $\hat\tau^*$ is given by
\begin{equation}
\label{eq:commit-foc}
-\int_{\underline\theta}^{\hat\theta^*} v(\theta)g(\theta)\mathrm{d}\theta\cdot f(-\hat\tau^*)+v(\hat\theta^*)g(\hat\theta^*)[F(-\hat\tau^*)-F(f^{-1}(\sigma/\hat\theta^*))]\cdot\widehat\Theta'(\hat\tau^*;\sigma)=0
\end{equation}
Recall that $\widehat\Theta'\geq0$ when $\hat\tau\geq0$ by \Cref{lemma:Theta-increasing}. And $F(-\hat\tau^*)-F(f^{-1}(\sigma/\hat\theta^*))<0$ when $\hat\tau\geq0$ and $=0$ when $\hat\tau<0$. Therefore, if $\hat\theta^*\leq\tilde\theta$, then the LHS of \Cref{eq:commit-foc} is strictly positive, leading to a contradiction. As a result, it must be $\hat\theta^*>\tilde\theta$. Also by the first-order condition, we must have $\hat\tau^*\geq0$. Since $\hat\theta^*(\sigma)=\widehat\Theta(\hat\tau^*(\sigma);\sigma)$ and $\tilde\theta\geq\hat\theta_-(\sigma)=\widehat\Theta(\hat\tau_-(\sigma);\sigma)$, it must hold that $\hat\tau^*(\sigma)>\hat\tau_-(\sigma)$.

That $\overline{V}\geq V$ is straightforward since the principal can always choose $\hat\tau=\hat\tau_-(\sigma)$. To see why $\overline{V}$ is decreasing in $\sigma$, by the Envelope Theorem, we have
\[
\overline{V}'(\sigma)=\int_{\hat\theta^*}^{\overline\theta} v(\theta)g(\theta)/\theta\frac{f(f^{-1}(\sigma/\theta))}{f'(f^{-1}(\sigma/\theta))}\mathrm{d}\theta+v(\hat\theta^*)g(\hat\theta^*)[F(-\hat\tau^*)-F(f^{-1}(\sigma/\hat\theta^*))]\frac{\partial\widehat\Theta}{\partial\sigma}.
\]
First, since we have $\hat\theta^*>\tilde\theta$, the first term is negative. Then, given that $\partial\widehat\Theta/\partial\sigma=\hat\theta/\sigma>0$, the second term is also negative. It completes the proof.
\end{proof}

{
\singlespacing
\bibliographystyle{chicago}
\bibliography{ref}
}

\newpage
\clearpage
\pagenumbering{arabic}
\renewcommand*{\thepage}{B-\arabic{page}}
\section{Supplemental Appendix}
\label{online-app}
\Cref{sect:asymptotics} records the proofs for the asymptotic results used in the proof of \Cref{prop:approval} part (b). \Cref{sect:optimistic} relaxes \Cref{as:ex-ante-bad} (the pessimistic prior) and \Cref{sect:bounded-noise} relaxes the assumption on the noise being unbounded in \Cref{as:noise}.

\subsection{Omitted Proofs for Asymptotics}
\label{sect:asymptotics}
\paragraph{Proof of \Cref{lm:b-linear}}
\begin{proof}
    Since $b(p)=1/\phi'(f^{-1}(p))$ and $\phi''(x)/(\phi'(x))^2\to0$, we have
\[
\lim_{p\searrow0}\frac{p\,b'(p)}{b(p)}=\lim_{p\searrow0}\frac{\phi''(f^{-1}(p))}{(\phi'(f^{-1}(p)))^2}=0.
\]
Therefore, for any fixed $c>0$, by dominated convergence theorem,
\[
\log\frac{b(cp)}{b(p)}=\int_p^{cp} \frac{b'(t)}{b(t)}\mathrm{d}t=\int_1^{c} \frac{(sp)b'(sp)}{b(sp)}\frac{1}{s}\mathrm{d}s\to0,\quad\text{as }p\to0.
\]
Therefore, $\lim_{p\to0}\frac{b(cp)-b(p)}{b(p)}=0$.
\end{proof}
\paragraph{Proof of \Cref{lm:f-inverse}}
\begin{proof}
Let $t:=-\log p$ and $s:=-\log c$. Then
\[
f^{-1}(cp)-f^{-1}(p)=\phi^{-1}(t+s)-\phi^{-1}(t)=s\frac{1}{\phi'(\phi^{-1}(t+\xi))}=sb(t+\xi),
\]
for some $\xi\in[0,s]$ by the mean value theorem. The conclusion then follows from the fact that
\[
\log\frac{\phi'(\phi^{-1}(t+\xi))}{\phi'(\phi^{-1}(t))}
=
\int_{\phi^{-1}(t)}^{\phi^{-1}(t+\xi)}\frac{\phi''(z)}{\phi'(z)}\,\mathrm{d}z\to0
\]
uniformly for bounded $\xi$.
\end{proof}
\paragraph{Proof of \Cref{lm:tail-bounds}}
\begin{proof}

We first establish two useful properties. First, $f$ has an exponential upper tail. Indeed, for any fixed $x_0>0$, by convexity,
\[
\phi(x)\ge \phi(x_0)+\phi'(x_0)(x-x_0)
\qquad\text{for all }x\ge x_0.
\]
Therefore, as $\phi'(x_0)>0$, there exist constants $c,C>0$ such that
\[
f(x)\le Ce^{-cx}
\qquad\text{for all large }x.
\]

Second, for any $\varepsilon>0$, there exists $C_\varepsilon>0$ such that
\[
\phi'(x)\le C_\varepsilon e^{\varepsilon\phi(x)}=C_\varepsilon (f(x))^{-\varepsilon}
\qquad\text{for all large }x.
\]
To see this, notice that $\phi''(x)/(\phi'(x))^2\to0$, so for all large $x$,
\[
\frac{\phi''(x)}{\phi'(x)}\le \varepsilon \phi'(x).
\]
Hence, for large enough $x_0>0$ and $x>x_0$,
\[
\log \phi'(x)-\log \phi'(x_0)
=
\int_{x_0}^x\frac{\phi''(z)}{\phi'(z)}\,\mathrm{d}z
\le
\varepsilon\int_{x_0}^x\phi'(z)\,\mathrm{d}z
=
\varepsilon[\phi(x)-\phi(x_0)].
\]
Thus, there exists $C_\varepsilon>0$ such that $\phi'(x)\le C_\varepsilon e^{\varepsilon\phi(x)}=C_\varepsilon (f(x))^{-\varepsilon}$ for large $x$.

Now let $x_p:=f^{-1}(p)$ for $p>0$, so that $\phi(x_p)=\log(1/p)$, $\phi'(x_p)=-f'(x_p)/p$, and when $p$ is small and thus $x_p$ is large, $\phi'(x_p)\le C_\varepsilon p^{-\varepsilon}$. As a result, as $p\to0$,
\[
\frac{f(t/p)}{p^2b(p)}
=
\frac{f(t/p)\phi'(x_p)}{p^2}
\le
C'e^{-ct/p}\cdot C_\varepsilon p^{-2-\varepsilon}\to0,
\]
which proves the first statement.

For the second statement, choose $\varepsilon=\frac12$. Then
\[
|f'(x)|=\phi'(x)f(x)\le C_{1/2} (f(x))^{1/2}=C' e^{-cx/2}
\qquad\text{for all large }x.
\]
Hence, as $p\to0$,
\[
\left|\frac{f'(t/p)}{p^3b(p)}\right|
=
\frac{|f'(t/p)|\,\phi'(x_p)}{p^3}
\le
C'e^{-ct/(2p)}\cdot C_\varepsilon p^{-3-\varepsilon}\to0.
\]
This completes the proof.
\end{proof}

\subsection{Optimistic Prior: $\mathbb{E}[v(\theta)]>0$}
\label{sect:optimistic}
This section provides equilibrium analysis when \Cref{as:ex-ante-bad} (pessimistic prior) does not hold. That is, what if $\mathbb{E}[v(\theta)]>0$?
\subsubsection{Moderate Optimism}
\begin{proposition}
\label{prop:moderate prior}
When $\mathbb{E}[v(\theta)]>0$ and $\mathbb{E}[v(\theta)/\theta]<0$, for small $\sigma$, there are only three kinds of equilibria: one pooling equilibrium $(-\infty,\overline\theta)$, an equilibrium $(\tau_-,\hat\theta_-)$ with $\hat\theta_-(\sigma)>\theta^\dagger$ and $\tau={T}(\hat\theta_-(\sigma);\sigma)$, and possibly some other equilibria.

Moreover, in the limit of $\sigma\to0$, the above equilibria either converge to the pooling one where all agent types pool on zero effort and the principal approves everyone, or converge to the equilibrium with $(\tau,\hat\theta)=(\theta^\dagger,\theta^\dagger)$.

\end{proposition}
\begin{proof}
    Again, the existence proof of the pooling equilibrium is trivial and thus omitted. For the existence of the equilibrium with $\hat\theta_-$, it is useful to first establish some properties of $R(\hat\theta)$. Let $\overline{\theta}{}':=\inf\{\hat\theta:\int_{\underline\theta}^{\overline\theta{}'}v(\theta)g(\theta)\mathrm{d}\theta\geq0\}$, then $\overline{\theta}{}'\in(\tilde\theta,\overline\theta)$. By definition, $R(\theta^\dagger)=0$ and $R(\hat\theta)<0$ for any $\hat\theta\in(\theta^\dagger,\overline{\theta}{}')$.
\begin{lemma}
    When $\mathbb{E}[v(\theta)]>0$ but $\mathbb{E}[v(\theta)/\theta]<0$, there exists $\theta_0\in(\theta^\dagger,\overline{\theta}{}')$ such that $R(\hat\theta)$ is strictly decreasing for $\hat\theta\in(\theta^\dagger,\min\{\tilde\theta,\theta_0\})$ and (weakly) decreasing for $\hat\theta\in(\max\{\tilde\theta,\theta_0\},\overline{\theta}{}')$.
\end{lemma}
\begin{proof}
    Recall that $R'(\hat\theta)=-v(\hat\theta)g(\hat\theta)[1/\hat\theta+R(\hat\theta)]/\big[\int_{\underline\theta}^{\hat\theta} v(\theta)g(\theta)\mathrm{d}\theta\big]^3$ where $\int_{\underline\theta}^{\hat\theta} v(\theta)g(\theta)\mathrm{d}\theta<0$. And also, $v(\tilde\theta)<0$ if and only if $\hat\theta<\tilde\theta$. It thus remains to show there exists $\theta_0\in(\theta^\dagger,\overline{\theta}{}')$ such that $1/\hat\theta+R(\hat\theta)>0$ if and only if $\hat\theta<\theta_0$. 

    Define
    \[
    M(\hat\theta):=\left[\int_{\underline\theta}^{\hat\theta} v(\theta)g(\theta)\mathrm{d}\theta\right][1/\hat\theta+R(\hat\theta)]=\frac{1}{\hat\theta}\int_{\underline\theta}^{\hat\theta} v(\theta)g(\theta)\mathrm{d}\theta+\int_{\hat\theta}^{\overline\theta} v(\theta)g(\theta)/\theta\mathrm{d}\theta.
    \]
    Notice that $M(\theta^\dagger)<0$ and $M(\overline{\theta}{}')>0$. In order to show $1/\hat\theta+R(\hat\theta)$ only crosses zero once from above, it is suffices to show that $M$ is strictly decreasing in $\hat\theta$, which is true because
    \[
    M'(\hat\theta)=-\frac{1}{\hat\theta^2}\int_{\underline\theta}^{\hat\theta} v(\theta)g(\theta)\mathrm{d}\theta<0.
    \]
    In conclusion, there exists $\theta_0\in(\theta^\dagger,\overline{\theta}{}')$ such that $1/\hat\theta+R(\hat\theta)>0$ if and only if $\hat\theta<\theta_0$. As a result, $R(\hat\theta)$ is strictly decreasing for $\hat\theta\in(\theta^\dagger,\min\{\tilde\theta,\theta_0\})$ and (weakly) decreasing for $\hat\theta\in(\max\{\tilde\theta,\theta_0\},\overline{\theta}{}')$.
\end{proof}
Notice that in this case, $\widehat{T}$ is only defined over $(\theta^\dagger,\overline\theta{}')$. Let $\tilde\theta':=\min\{\tilde\theta,\theta_0\}$. Then $\widehat{T}$ is strictly decreasing over $(\theta^\dagger,\tilde\theta']$ with the same property as in our baseline model. 
Since $\widehat\Theta$ remains unchanged, all previous arguments follow so that when $\sigma$ is small, $\widehat\Theta\circ\widehat{T}$ has a unique fixed point in $(\theta^\dagger,\tilde\theta']$, denoted by $\hat\theta_-(\sigma)$. As $\sigma$ converges to 0, $\hat\theta_-(\sigma)$ converges to $\theta^\dagger$. 

Of course, there can be other equilibria as fixed points of $\widehat\Theta\circ\widehat{T}$ above $\tilde\theta'$ or fixed points of $\widehat\Theta\circ(-\widehat{T})$. However, we show that they will either vanish or converge to the pooling equilibrium in the limit of $\sigma\to0$. First, for any fixed point $\hat\theta$ of $\widehat\Theta\circ(-\widehat{T})$, we must have $1/\hat\theta+R(\hat\theta)=0$, i.e., $\hat\theta=\theta_0$ because $\widehat\Theta((-\widehat{T}(\hat\theta;\sigma);\sigma)=\inf\{\theta:-f^{-1}(-\sigma R(\hat\theta))+f^{-1}(\sigma/\theta)\geq0\}$. Note that the fixed point $\hat\theta=\theta_0$ is independent of $\sigma$. Therefore, in the limit of $\sigma\to0$, we will have $(\tau,\hat\theta)=(-\sigma f^{-1}(-\sigma R(\hat\theta)),\hat\theta)$ converge to $(0,\theta_0)$. In this limit equilibrium, all types below $\theta_0$ exert zero effort, while all types above $\theta_0$ exert zero effort as well since $\lim_{\sigma\to0}\tau+\sigma f^{-1}(\sigma/\theta)=0$. Hence, this limit equilibrium is exactly the pooling equilibrium.

Then, consider the fixed point(s) of $\widehat\Theta\circ\widehat{T}$ above $\tilde\theta'$, if it exists. For any $\hat\theta\in[\tilde\theta',\overline\theta{}')$, we have $R(\hat\theta)\leq\sup_{\hat\theta\in[\tilde\theta',\max\{\tilde\theta,\theta_0\}]}R(\hat\theta)=:R<0$ by the previous lemma. Therefore, for any sequence of $\sigma_k$ and equilibrium $(\tau_k,\hat\theta_k)$ under $\sigma_k$ such that $\sigma_k\to0$ (if such a sequence exists), we have
    \[
    0\leq\lim_{k\to\infty}\tau_{k}=\lim_{n\to\infty}T(\hat\theta_{k};\sigma_{k})\leq\lim_{n\to\infty}\sigma_{k}f^{-1}(-\sigma_{k}R)=0,
    \] 
    while $\lim_{k\to\infty}\tau_{k}/\sigma_{k}=\lim_{k\to\infty}\widehat{T}(\hat\theta_{k};\sigma_{k})=\lim_{k\to\infty}f^{-1}(-\sigma_k R(\hat\theta_{k}))\geq0$.
    
    However, this leads to a violation of the indifference condition of type $\hat\theta_{k}$:  
    \begin{align*}
        \lim_{k\to\infty}& \left\{F(f^{-1}(\sigma_{k}/\hat\theta_{k}))-[\tau_{k}+\sigma_{k}f^{-1}(\sigma_{k}/\hat\theta_{k})]/\hat\theta_{k}\right\}=1\\
        &>F(0)\geq\lim_{k\to\infty} F(-\tau_{k}/\sigma_{k})=\lim_{k\to\infty} F_{\sigma_{k}}(0-\tau_{k}).
    \end{align*}
    As a result, when $\sigma$ is vanishing, $\widehat\Theta\circ\widehat{T}$ has no fixed point above $\tilde\theta'$.
\end{proof}
According to \Cref{prop:moderate prior}, when $\sigma$ is small, the equilibrium with $\hat\theta_-(\sigma)$ is principal-optimal. The comparative statics about $\hat\theta_-(\sigma)$ and $V(\sigma)$ are the same as what we presented in the baseline model because the property of $\widehat{T}$ to the left of $\tilde\theta$ remains the same as before. Hence \Cref{prop:non-monotone} and thus the pitfall of precision still hold.

\subsubsection{Extreme Optimism}
\begin{proposition}
    When $\mathbb{E}[v(\theta)/\theta]\geq0$, only the pooling equilibrium exists where all agent types pool on zero effort and the principal approves everyone.
\end{proposition}
\begin{proof}
    Again any interior equilibrium must be a fixed point of $\widehat\Theta\circ\widehat{T}$ or $\widehat\Theta\circ(-\widehat{T})$. Remember that $\widehat{T}(\hat\theta;\sigma)=f^{-1}(-\sigma R(\hat\theta))$ is only defined over $\hat\theta$ such that $R(\hat\theta)\leq0$. When $\mathbb{E}[v(\theta)/\theta]\geq0$, we have the following result on $R(\hat\theta)$:
\begin{lemma}
    When $\mathbb{E}[v(\theta)/\theta]\geq0$, if $R(\hat\theta)\leq0$ for some $\hat\theta\in(\underline\theta,\overline\theta)$, $1/\hat\theta+R(\hat\theta)<0$.
\end{lemma}
\begin{proof}
    Given that $\mathbb{E}[v(\theta)/\theta]\geq0$, $\int_{\hat\theta}^{\overline\theta} v(\theta)g(\theta)/\theta\mathrm{d}\theta\geq0$ for any $\hat\theta$, therefore $R(\hat\theta)\leq0$ if and only if $\int_{\underline\theta}^{\hat\theta} v(\theta)g(\theta)\mathrm{d}\theta<0$. Let $\theta':=\inf\big\{\hat\theta\in[\underline\theta,\overline\theta]:\int_{\underline\theta}^{\hat\theta} v(\theta)g(\theta)\mathrm{d}\theta\geq0\big\}$. Then $R(\hat\theta)\leq0$ if and only if $\hat\theta\in[\underline\theta,\theta')$ and $\theta'\geq\tilde\theta$.
    
    Observe that when $\hat\theta\in(\underline\theta,\tilde\theta]$,
    \begin{align*}
        \frac{1}{\hat\theta}\int_{\underline\theta}^{\hat\theta} v(\theta)g(\theta)\mathrm{d}\theta+\int_{\hat\theta}^{\overline\theta} v(\theta)g(\theta)/\theta\mathrm{d}\theta>&\int_{\underline\theta}^{\hat\theta} v(\theta)g(\theta)/\theta\mathrm{d}\theta+\int_{\hat\theta}^{\overline\theta} v(\theta)g(\theta)/\theta\mathrm{d}\theta\\
        =&\mathbb{E}[v(\theta)/\theta]\geq0.
    \end{align*}
    When $\hat\theta\in(\tilde\theta,\theta')$, instead
    \begin{align*}
        \frac{1}{\hat\theta}\int_{\underline\theta}^{\hat\theta} v(\theta)g(\theta)\mathrm{d}\theta+&\int_{\hat\theta}^{\overline\theta} v(\theta)g(\theta)/\theta\mathrm{d}\theta>\frac{1}{\tilde\theta}\int_{\underline\theta}^{\hat\theta} v(\theta)g(\theta)\mathrm{d}\theta+\int_{\hat\theta}^{\overline\theta} v(\theta)g(\theta)/\theta\mathrm{d}\theta\\
        =&\frac{1}{\tilde\theta}\int_{\underline\theta}^{\tilde\theta} v(\theta)g(\theta)\mathrm{d}\theta+\frac{1}{\tilde\theta}\int_{\tilde\theta}^{\hat\theta} v(\theta)g(\theta)\mathrm{d}\theta+\int_{\hat\theta}^{\overline\theta} v(\theta)g(\theta)/\theta\mathrm{d}\theta\\
        >&\int_{\underline\theta}^{\tilde\theta} v(\theta)g(\theta)/\theta\mathrm{d}\theta+\int_{\tilde\theta}^{\hat\theta} v(\theta)g(\theta)/\theta\mathrm{d}\theta+\int_{\hat\theta}^{\overline\theta} v(\theta)g(\theta)/\theta\mathrm{d}\theta\\
        =&\mathbb{E}[v(\theta)/\theta]\geq0.
    \end{align*}
    As a result, in either case, we have $1/\hat\theta+R(\hat\theta)<0$.
\end{proof}
    Therefore, $\widehat{T}(\hat\theta;\sigma)<f^{-1}(\sigma/\hat\theta)$ and thus
    \[
    \widehat\Theta(\widehat{T}(\hat\theta;\sigma);\sigma)<\widehat\Theta(f^{-1}(\sigma/\hat\theta);\sigma)\leq\hat\theta
    \]
    and similarly,
    \[
    \widehat\Theta(-\widehat{T}(\hat\theta;\sigma);\sigma)<\widehat\Theta(-f^{-1}(\sigma/\hat\theta);\sigma)=\hat\theta.
    \]
    As a result, $\widehat\Theta\circ\widehat{T}$ and $\widehat\Theta\circ(-\widehat{T})$ have no fixed point over $\hat\theta$ such that $R(\hat\theta)\leq0$. Accordingly, there is no interior equilibrium and the only (corner) equilibrium is the pooling one.
\end{proof}

\subsection{Bounded Noise}
\label{sect:bounded-noise}
    The equilibrium analysis does not rely on the noise being bounded. We only utilize the unboundedness in the proofs of \Cref{prop:non-monotone,prop:error,prop:approval}. This section shows how these proofs can extend to bounded noise.
    
    With bounded noise, there exists $\underline\sigma>0$ such that $\hat\theta_-(\sigma)=\theta^\dagger$ and $f(0-\hat\tau_-(\sigma))=0$ for all $\sigma\in(0,\underline\sigma]$. We show that the results in \Cref{prop:non-monotone} still hold. Similar arguments can be used to adapt the proofs of \Cref{prop:error,prop:approval}.
    
    \begin{proposition}[Proposition 3']
    $V'(\sigma)>0$ for $\sigma\in(0,\underline\sigma]$. 
    \end{proposition}
    \begin{proof}
        When $\sigma<\tilde{\sigma}$, $A_3\geq0$ since $\mathrm{d}\hat\theta_-/\mathrm{d}\sigma\geq0$ by \Cref{lemma:increasing-indifference}. And for $\sigma\in(0,\underline\sigma]$, $A_1=0$ given $f(0-\hat\tau_-(\sigma))=0$, and 
    \[
    A_2=\int_{\theta^\dagger}^{\overline\theta} v(\theta)g(\theta)/\theta\frac{f(f^{-1}(\sigma/\theta))}{f'(f^{-1}(\sigma/\theta))}\mathrm{d}\theta>0,
    \]
    provided that $\int_{\theta^\dagger}^{\overline\theta} v(\theta)g(\theta)/\theta\mathrm{d}\theta=0$ and $f'/f$ is decreasing.\footnote{Note that for bounded noise, it is impossible to have constant $\frac{f'(x)}{f(x)}$ when $f(x)$ is close to zero, hence $A_2\ne0$.} As a result, $V'(\sigma)>0$ for $\sigma\in(0,\underline\sigma]$.
    \end{proof}

\end{document}